\documentclass[a4paper, 11pt]{article}
\usepackage[left=1.5cm, right=1cm, top=0.5cm, bottom=1cm, includeheadfoot]{geometry}
\usepackage[english]{babel}
\usepackage[utf8]{inputenc}
\usepackage[authoryear,round]{natbib}
\usepackage{framed}
\usepackage{amsmath,amsfonts,amssymb,amsthm,mathtools,bbm}
\usepackage{booktabs, multirow}
\usepackage{adjustbox}
\usepackage{float}
\usepackage{listings}
\usepackage{color}
\usepackage{url}
\usepackage{hyperref}
\usepackage{xpatch}
\usepackage{xcolor}
\usepackage{listings}
\usepackage{realboxes}
\usepackage{marginnote}
\usepackage{changepage}
\usepackage{fancyhdr}
\usepackage{algorithm, algpseudocode}
\usepackage{afterpage}
\usepackage{tikz}
\usetikzlibrary{bayesnet, decorations.pathmorphing} 
\usepackage[many]{tcolorbox}

\newtheorem{theorem}{Theorem}
\newtheorem{proposition}{Proposition}
\newtheorem{lemma}{Lemma}
\newtheorem{remark}{Remark}
\newtheorem{corollary}{Corollary}
\newtheorem{conjecture}[lemma]{Lemma}

\date{}

\newcommand{\orcid}[1]{\href{https://orcid.org/#1}{\includegraphics[height=\fontcharht\font`A]{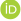}}}

\hypersetup{
	colorlinks=true,
	linkcolor=blue,
	citecolor=magenta,   
	urlcolor=cyan,
	pdfpagemode=FullScreen,
}

\DeclareMathOperator{\EX}{\mathbb{E}}
\DeclareMathAlphabet{\mathpzc}{OT1}{pzc}{m}{it}
\DeclareRobustCommand{\stirling}{\genfrac[]{0pt}{}}
\DeclareRobustCommand{\stirlings}{\genfrac\{\}{0pt}{}}
\definecolor{applegreen}{rgb}{0.55, 0.71, 0.0}
\definecolor{darkolivegreen}{rgb}{0.33, 0.42, 0.18}
\definecolor{codecolors}{RGB}{ 219,225,226}
\definecolor{codecolorsinline}{RGB}{ 230.1000,  234.3000 , 235.0000}

\tcolorboxenvironment{lstlisting}{
	spartan,
	colframe=black,
	left=1mm,right=1mm,top=-2mm,bottom=-2mm,
	colback=codecolors
}

\lstdefinestyle{code}{
	backgroundcolor = \color{codecolors},
	mathescape=true,
	keywordstyle=\color{blue}\bf,
	commentstyle=\color{darkolivegreen},
	stringstyle=\color{darkolivegreen},
	emphstyle=\color{magenta},	
	basicstyle=\footnotesize\ttfamily,
	numbers=left, 
	breaklines=true
}
\lstdefinestyle{codeinput}{
	mathescape=true,
	keywordstyle=\color{blue}\bf,
	commentstyle=\color{darkolivegreen},
	stringstyle=\color{darkolivegreen},
	emphstyle=\color{magenta},	
	basicstyle=\footnotesize\ttfamily,
	breaklines=true
}

\lstdefinestyle{codeinputinline}{
	backgroundcolor = \color{codecolorsinline},
	mathescape=true,
	keywordstyle=\color{blue}\bf,
	commentstyle=\color{green},
	stringstyle=\color{darkolivegreen},
	emphstyle=\color{magenta},	
	basicstyle=\normalsize\ttfamily,
}

\lstdefinestyle{codeoutput}{
	basicstyle=\footnotesize\ttfamily,
	mathescape=true,
}

\makeatletter
\renewcommand{\env@cases}[1][@{}l@{\quad}l@{}]{%
	\let\@ifnextchar\new@ifnextchar
	\left\lbrace
	\def\arraystretch{1.2}%
	\array{#1}%
}
\makeatother

\newcommand{\factoredgeselect}[4][]{ %
	\foreach \f in {#3} { %
		\foreach \x in {#2} { %
			\path (\x) edge[-,#1] (\f) ; 
		} ;
		\foreach \y in {#4} { %
			\path (\f) edge[-|, >={triangle 45}, decorate,decoration={snake, amplitude=1mm, segment length=5mm}, #1] (\y) ; 
		} ;
	} ;
}

\title{Exact computation of posterior distribution of mixture weights in hierarchical Bayesian models}
\author{
	\begin{tabular}{c}
	Georgy Meshcheryakov \\
		\hyperlink{mailto:iam@georgy.top}{iam@georgy.top}
	\end{tabular}
}
\begin{document}

\maketitle

\begin{abstract}
	Hierarchical mixture models are a powerful tool for modeling data generated from heterogeneous sources, particularly when the mixing proportion $\boldsymbol{w}$ itself is treated as a random variable with a Dirichlet or Beta-Liouville prior. Such models are widely employed in scenarios where uncertainty in class membership or data-generating processes must be probabilistically quantified. This paper studies the exact marginalization of the mixture weight \( \boldsymbol{w} \sim \mathpzc{BL}(\boldsymbol{\alpha}, \gamma) \). For the two-component case we give an $O(n^2)$ dynamic program -- and an $O(n \log^2 n)$ FFT variant -- for the marginal likelihood, and show that the exact posterior of the weight is a finite mixture of Beta distributions, delivering closed-form posterior summaries, credible intervals and per-observation local false-discovery rates without any sampling. For $K \ge 3$ components we give an exact joint dynamic program. The closed form also yields structural guarantees that approximations cannot match: the posterior over the latent count of non-null observations is log-concave and unimodal (a consequence of Newton's inequalities for the elementary symmetric polynomials), its moments are exact ratios of the marginal likelihood at shifted hyperparameters, and it is stochastically monotone in the data. The algorithms are exact and deterministic, agreeing with brute-force, arbitrary-precision and Monte-Carlo references to round-off. Their practical value is not accuracy over a sampler -- which targets the same posterior -- but that the closed-form posterior is \emph{calibrated} in the small-sample, rare-signal regime where cheap approximations are not (across our simulations the exact $95\%$ credible interval holds nominal coverage while a logit-Laplace interval drops to $78\%$), and that the empirical-Bayes evidence is a smooth, deterministic objective rather than a noisy Monte-Carlo one. It reproduces MCMC posteriors without tuning or convergence diagnostics, at a fraction of the cost for small-to-moderate $n$. The gain is largest in the small-sample regime the method is built for: on a real multilevel meta-analysis, a pathway-level dysregulation analysis of leukemia gene expression, and a leukemia-derived gene-panel benchmark with known ground truth, the exact interval for the signal proportion is calibrated where EM gives no interval at all (collapsing to a boundary) and Gaussian/Laplace approximations mis-cover, and it is two orders of magnitude faster than the sampler that would match it. On the large prostate-cancer benchmark, where every method has ample data, it agrees with \texttt{locfdr} on the gene ranking while adding a posterior interval for the null proportion.
\end{abstract}

\section{Introduction}

\begin{figure}[t]
	\centering
	\begin{tikzpicture}
		
		\node[const] (a) {\( \boldsymbol{\alpha} \)}; 
		\node[const, right=of a] (b) {\( \gamma \)};

    	\factor[below=of a, xshift=0.5cm] {Beta} {right:$\mathpzc{BL}$} {} {} ; %
		\node[latent, below=of Beta] (w) {$\boldsymbol{w}$}; 
		
		\factor[right=of w] {Cat} {below:$\mathpzc{Cat}$} {} {};
		\node[const, right=of Cat] (FG) {\( [\mathpzc{F}, \mathpzc{G}] \)}; 
		
		\node[obs, right=of FG] (x) {\( x \)}; 
		
		\plate {subpop} {(x) (Cat) (FG)} {\( n = 1 \dots N_m \)}; 
		\plate {pop} {(subpop) (w) (Beta)} {\( m = 1 \dots M \)};		

		\factoredge {a,b} {Beta} {w};
		\factoredgeselect {w} {Cat} {FG};
		\edge {FG}  {x};

	\end{tikzpicture}

	\caption{Model diagram in the plate notation.}
	\label{fig:model}
\end{figure}
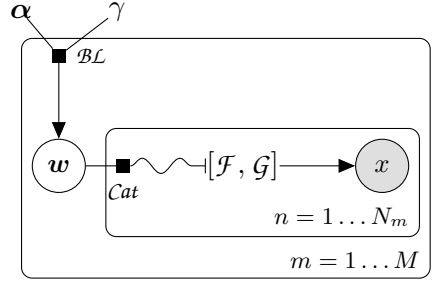

Assume that we are dealing with a stratified dataset of $M$ subpopulations/stratas $\mathpzc{D} = \left\{ \mathpzc{X}^{(i)}\right\}_{i=1}^M$, where $i$-th "strata" has a size of $N_i$: $\mathpzc{X}^{(i)} = \{x_k\}_{k=1}^{N_i}$. For each strata, for e.g. for $\mathpzc{X}_i$, it it is known that a data sample $x$ from $i$-th strata orginated either from $\mathpzc{F}$ of $\mathpzc{G}$ distributions with a probability of the former being $w_i $. That's it,  for $i$-th strata, 
$$ \mathpzc{X}^{(i)} \sim \mathpzc{Mixture}(\mathpzc{F}, \mathpzc{G}, w_i), w_i \sim \mathpzc{Beta}$$
where mixture is defined in such a way that a PDF of each sample is simply $$h(x|\theta) = w_i f(x|\theta) + (1 - w_i) g(x|\theta),$$
where $\theta$ is some arbitrary parameter vector, $f$ is the PDF of $\mathpzc{F}$, $g$ is the PDF of $\mathpzc{G}$. Or, in a plate-notation, see \autoref{fig:model}. Such models frequently arise in different fields, for example

\begin{enumerate}
 \item \textbf{Multiple Testing and False Discovery Rate (FDR) Estimation}: Two-component Beta mixtures are extensively used in high-dimensional hypothesis testing, where one component represents the null distribution (often uniform or a known parametric form) and the other captures alternative hypotheses \citep{example_fdr_2012};
 \item \textbf{A/B testing and p-curve mixture models:} closely related to the previous example, these models explicitly separate items into latent subgroups (e.g., responders vs. non-responders) and estimate both prevalence $w$ and subgroup-specific effects \citep{example_pcurve_2025}. There, $\mathpzc{G}$ is null-effect distribution and $\mathpzc{F}$ is some \textit{success} distribution.
\end{enumerate}

There are two most popular methods to deal with mixture weights, Expected Maximization (EM) and MCMC-based approaches, both harboring their own drawbacks. Namely, the former: 

\begin{itemize}
	\item The classical degeneracies of EM for mixtures -- an unbounded log-likelihood and sharp likelihood ridges, mitigated by ad-hoc corrections such as rescaling \([0,1]\) to \([\varepsilon,1-\varepsilon]\) \citep{Schrder2017,Dias2004,example_methylation_2024} -- arise when the component parameters are estimated \emph{jointly} with the weight. They do not affect the fixed-component, weight-only problem studied here, whose profile objective in $w$ is well behaved; we make no claim of superiority on this ground;
	\item Lack of uncertainty quantification: EM returns a point estimate of $w$ and no posterior, so it yields no credible interval for the mixing proportion and cannot propagate that uncertainty into model comparison or into an interval for the global null fraction $\pi_0$. When the true weight sits at a boundary ($0$ or $1$) the estimate is a point mass with no interval at all. (Cheap Bayesian approximations recover an interval but, as we show in \autoref{sec:numerical}, a Gaussian/Laplace approximation is badly miscalibrated in exactly the small-sample, rare-signal regime of interest.)
\end{itemize}

As for MCMC:
\begin{itemize}
	\item The computational complexity of MCMC grows quadratically with the number of stratas and observations, as each strata introduces an extra parameter and each iteration requires reevaluating the likelihood for all data points. In high-throughput applications such as large-scale multiple testing with tens of thousands of hypotheses  or p-curve meta-analyses with thousands of studies, MCMC is rendered prohibitively slow;
	\item Storing posterior samples for	$w$ and other latent variables across thousands of iterations becomes infeasible for large $M$ and $N_i$;
	\item Hierarchical models with strongly correlated parameters (e.g., \( w \) and component parameters \( \theta_F, \theta_G \)) suffer from poor mixing and slow convergence. Diagnostics like effective sample size (ESS) or Gelman-Rubin statistics often reveal inadequate exploration of the posterior, especially when \( w \) is near 0 or 1.
\end{itemize}
 These limitations call for algorithms that tackle both problems at once: quantifying the uncertainty in $w$ and scaling to many strata. Consider a generic two-component classification task: a practitioner must simultaneously infer the posterior distribution of \( w \) to quantify the global proportion of one component and assign item-specific probabilities to guide downstream decisions to assess robustness.

The two-component instance of this model, with a known null component, is the \emph{two-groups model} of large-scale inference \citep{Efron2008}: there the per-observation posterior null probability is the local false discovery rate and the global null fraction $\pi_0$ is the complement of the mixing weight. A substantial literature estimates these quantities -- Efron's empirical-null and \texttt{locfdr} machinery \citep{Efron2004,Efron2007}, Storey's $\pi_0$ estimator \citep{Storey2002}, Strimmer's \texttt{fdrtool} \citep{Strimmer2008}, the semiparametric hierarchical mixture of \citet{Newton2004}, and adaptive-shrinkage empirical Bayes \citep{Stephens2017} -- but each returns point estimates or tail areas. What has been missing, and what we supply, is the \emph{exact} posterior of the mixing weight itself: a closed-form finite mixture of Beta distributions that yields calibrated credible intervals for the proportion and per-observation local FDRs in a single deterministic pass, without sampling.

 We start formulating the problem by writing down the joint likelihood of the dataset $\mathpzc{D}$ and mixture weights vector $\mathbf{w} = \{w_i\}_{i=1}^M$:
$$\widetilde{L}(\mathpzc{D}, \mathbf{w}|\theta) = \prod_{i=1}^M L(\mathpzc{X}^{(i)}, \mathbf{w}|\theta),$$
where $L(\mathpzc{X}^{(i)}, \mathbf{w}|\theta)$ is a sub-likelihood for the $i$-th strata, $\mathbf{\theta}$ is some parameters vector. Note that the if we treat the model in a MAP way, number of parameters grows linearly with a number of stratas, and that becomes a problem when $M$ is large. 

Instead of using joint likelihood as is, we seek to find the marginal likelihood (that, in turn, would allow us to easily compute the posterior distribution of $\mathbf{w}$ due to the Bayes theorem):
$$ \widetilde{L}(\mathpzc{D}|\theta) = \int \widetilde{L}(\mathpzc{D}|\theta, \mathbf{w}) q(\mathbf{w}|\alpha, \beta) d\mathbf{w} =  \prod_{i=1}^M \int_{0}^{1}\widetilde{L}(\mathpzc{X}^{(i)}|\theta, w_i) q(w_i|\alpha, \beta) dw_i,$$
where $q(w_i|\alpha, \beta) = \frac{1}{B(\alpha, \beta)} w^{\alpha - 1} (1 - w)^{\beta - 1}$ is PDF of the $\mathpzc{Beta}(\alpha, \beta)$ distribution.
This paper is dedicated to algorithms to compute integrals $\int_{0}^{1}\widetilde{L}(\mathpzc{X}^{(i)}|\theta, w_i) q(w_i|\alpha, \beta) dw_i$. The roadmap of the paper is as follows:

\begin{enumerate}
	\item We derive a specialized $O(n^2)$ dynamic-programming (DP) algorithm for the 2-component mixture with a $\mathpzc{Beta}$-distributed weight, evaluated stably in the log-domain; a Fast-Fourier-Transform pairing scheme attains a lower $O(n \log^2 n)$ complexity but, operating necessarily in the linear domain, is numerically reliable only for small $n$ and is reported as a complexity result, not a recommended backend;
	\item We show that the exact posterior of the weight is a finite mixture of Beta distributions, yielding closed-form posterior means, variances, credible intervals and per-observation posterior class probabilities (local FDR) without sampling;
	\item We establish three structural properties of this posterior that its fast approximations do not share: the posterior over the latent count of non-null observations is log-concave and unimodal by Newton's inequalities, its moments are exact ratios of the marginal likelihood at shifted hyperparameters, and it is stochastically monotone in the data;
	\item We derive the corresponding algorithm for $K$-component mixtures with a $\mathpzc{BL}$-distributed weight vector, where the multivariate coefficient array does not factor across components, giving an exact joint dynamic program of cost $O(n^{K-1})$ in storage;
	\item We validate all algorithms against brute-force, quadrature, arbitrary-precision and Monte-Carlo references, benchmark them against EM and MCMC, and demonstrate the method on the standard prostate-cancer microarray data and in a hierarchical multi-study setting.
\end{enumerate}

\section{The special case of 2-component mixtures}\label{sec:twocomp}
Without a loss of generality and for brevity purposes, we consider a particular strata $\mathpzc{X}^{(j)}$ denoted as $\mathpzc{X}$ with a number of observations $n$, and the only mixture parameter be denoted as $w$.
Here, we start off by analyzing a special case of $K = 2$, where $w \sim \mathpzc{Beta}(\alpha, \beta)$. Let's write down the integral of interest:

\begin{equation*}
	\begin{split}
	\int_{0}^{1}\widetilde{L}(\mathpzc{X}|\theta, \mathbf{w}) q(\mathbf{w}|\alpha, \beta) dw = \int_{0}^{1} \left(\prod_{i=1}^n (w f(x_i|\theta) + (1 - w)g(x_i|\theta)) \right) \overbrace{\frac{w^{\alpha - 1} (1 -w)^{\beta - 1}}{B(\alpha, \beta)}}^{q(w|\alpha, \beta)} dw = \\ = \underbrace{\frac{1}{B(\alpha, \beta)}\left(\prod_{i=1}^n f(x_i|\theta)  \right)}_{L_0(\mathpzc{X}|\theta)} \int_{0}^{1} w^n \left(\prod_{i=1}^n \left(1 + \frac{(1 - w)}{w}\underbrace{\frac{g(x_i|\theta)}{f(x_i|\theta)}}_{t_i}\right) \right) {w^{\alpha - 1} (1 -w)^{\beta - 1}} dw = \\ = L_0(\mathpzc{X}|\theta) \int_{0}^{1} \left(\prod_{i=1}^n \left(1 + \frac{1-w}{w} t_i\right) \right) w^{n + \alpha - 1} (1 - w)^{\beta - 1} dw.
	\end{split}
\end{equation*}
The product term can be expanded into a $n$-variate sum from $0$ to $1$:
\begin{equation}
\prod_{i=1}^n \left(1 + \frac{1-w}{w}t_i\right) = \sum_{k_1 = 0}^1 \sum_{k_2=0}^1  \hdots \sum_{k_n = 0}^1 \left(\frac{1-w}{w}\right)^{\sum_{i=1}^n k_{i}} \prod_{i=1}^n t_{i}^{k_{i}} = \sum_{\mathclap{\mathbf{0}_n \le \mathbf{k} \le \mathbf{1}_n}} \left(\frac{1-w}{w}\right)^{|\mathbf{k}|} \mathbf{t}^\mathbf{k}, \label{eq:polynomial}
\end{equation}
In the last step, we switched to the \href{https://en.wikipedia.org/wiki/Multi-index_notation}{multi-index notation}. Then, the integral and the sum are exchanged and the famous Euler integral of the first kind is easily recognized:
\begin{equation*}
	\begin{split}
	 L_0(\mathpzc{X}|\theta) \int_{0}^{1} \sum_{\mathclap{\mathbf{0}_n \le \mathbf{k} \le \mathbf{1}_n}} \mathbf{t}^\mathbf{k} \left(\frac{1-w}{w}\right)^{|\mathbf{k}|}  w^{n + \alpha - 1} (1 - w)^{\beta - 1} dw  = 
	  L_0(\mathpzc{X}|\theta) \int_{0}^{1} \sum_{\mathclap{\mathbf{0}_n \le \mathbf{k} \le \mathbf{1}_n}} \mathbf{t}^\mathbf{k}   w^{n - |\mathbf{k}| + \alpha - 1} (1 - w)^{|\mathbf{k}| + \beta - 1} dw = \\ =
	   L_0(\mathpzc{X}|\theta) \sum_{\mathclap{\mathbf{0}_n \le \mathbf{k} \le \mathbf{1}_n}} \mathbf{t}^\mathbf{k}   \int_{0}^{1} w^{n - |\mathbf{k}| + \alpha - 1} (1 - w)^{|\mathbf{k}| + \beta - 1} dw =  L_0(\mathpzc{X}|\theta) \sum_{\mathclap{\mathbf{0}_n \le \mathbf{k} \le \mathbf{1}_n}} \mathbf{t}^\mathbf{k}  B(n - |\mathbf{k}| + \alpha, |\mathbf{k}| + \beta) = \\
	   = L_0(\mathpzc{X}|\theta) \sum_{\mathclap{\mathbf{0}_n \le \mathbf{k} \le \mathbf{1}_n}} \mathbf{t}^\mathbf{k} \frac{\Gamma(n - |\mathbf{k}| + \alpha) \Gamma(|\mathbf{k}| + \beta)}{\Gamma(n + \alpha + \beta)}=  \frac{L_0(\mathpzc{X}|\theta)}{\Gamma(n + \alpha + \beta)} \underbrace{\sum_{\mathclap{\mathbf{0}_n \le \mathbf{k} \le \mathbf{1}_n}} \mathbf{t}^\mathbf{k} \Gamma(n - |\mathbf{k}| + \alpha) \Gamma(|\mathbf{k}| + \beta)}_{S_{\mathbf{t}}(n, \alpha, \beta)} = \\= \frac{L_0(\mathpzc{X}|\theta)}{\Gamma(n + \alpha + \beta)}  S_{\mathbf{t}}(n, \alpha, \beta)
	\end{split}
\end{equation*}

There, the function $S_{\mathbf{t}}(n, \alpha, \beta)$ is defined so that its first argument doubles as the dimension of the summation. For $j\le n$,

\begin{equation}
S_{\mathbf{t}}(j, \alpha, \beta) = \sum_{\mathclap{\mathbf{0}_j \le \mathbf{k} \le \mathbf{1}_j}} \mathbf{t}_{:j}^\mathbf{k} \Gamma(j - |\mathbf{k}| + \alpha) \Gamma(|\mathbf{k}| + \beta), \label{eq:combfun}
\end{equation}

where $\mathbf{t}_{:j}$  is a vector of length $j$ formed from first $j$ elements $\mathbf{t}$ (if $j=n$ then $\mathbf{t}_{:n} = \mathbf{t}$).

The straightforward computation of $S_{\mathbf{t}}(n, \alpha, \beta)$ requires computing $\color{red}{2^n}$ terms. That's infeasible, but we can reduce the complexity of the problem significantly. 
\begin{proposition}
	The following recurrence equation holds:
	\begin{equation}\label{eq:marg_rec1}
	S_{\mathbf{t}}(n, \alpha, \beta) = S_{\mathbf{t}}(n - 1, \alpha + 1, \beta) + t_n S_{\mathbf{t}}(n - 1, \alpha, \beta + 1),~\text{for all }n \ge 1,
	\end{equation}

	with the base case $S_{\mathbf{t}}(0, \alpha, \beta) = \Gamma(\alpha)\Gamma(\beta)$, where $t_n = g(x_n\mid\theta)/f(x_n\mid\theta)$ is the $n$-th likelihood ratio.
\end{proposition}
\begin{proof}
	Let's unfold the sum in $S_{\mathbf{t}}(n, \alpha, \beta)$ with respect to the latest summation operation:
	\begin{equation*}
		\begin{split}
			S_{\mathbf{t}}(n, \alpha, \beta) = \sum_{\mathclap{\mathbf{0}_n \le \mathbf{k} \le \mathbf{1}_n}} \mathbf{t}^\mathbf{k} \Gamma(n - |\mathbf{k}| + \alpha) \Gamma(|\mathbf{k}| + \beta) = \sum_{\mathclap{\mathbf{0}_{n - 1} < \mathbf{k} < \mathbf{1}_{n - 1}}} ~~~~~~\sum_{k_n=0}^1 \mathbf{t}^\mathbf{k} t_n^{k_n} \Gamma(n - |\mathbf{k}| - k_n + \alpha) \Gamma(|\mathbf{k}| + k_n + \beta) = \\
			= \sum_{\mathclap{\mathbf{0}_{n - 1} < \mathbf{k} < \mathbf{1}_{n - 1}}} \mathbf{t}^\mathbf{k}  \Gamma(n - |\mathbf{k}| + \alpha ) \Gamma(|\mathbf{k}| + \beta)   + t_n \mathbf{t}^\mathbf{k}  \Gamma(n - 1 - |\mathbf{k}| + \alpha ) \Gamma(|\mathbf{k}| + \beta + 1) = \\ = \sum_{\mathclap{\mathbf{0}_{n - 1} < \mathbf{k} < \mathbf{1}_{n - 1}}} \mathbf{t}^\mathbf{k}  \Gamma(n - |\mathbf{k}| + \alpha + 1 - 1) \Gamma(|\mathbf{k}| + \beta) +  t_n \sum_{\mathclap{\mathbf{0}_{n - 1} < \mathbf{k} < \mathbf{1}_{n - 1}}} \mathbf{t}^\mathbf{k}  \Gamma(n - 1 - |\mathbf{k}| + \alpha ) \Gamma(|\mathbf{k}| + \beta + 1) = \\ = \sum_{\mathclap{\mathbf{0}_{n - 1} < \mathbf{k} < \mathbf{1}_{n - 1}}} \mathbf{t}^\mathbf{k}  \Gamma((n - 1) - |\mathbf{k}| + (\alpha + 1)) \Gamma(|\mathbf{k}| + \beta) + t_n S_{\mathbf{t}}(n - 1, \alpha, \beta + 1) = \\ =S_{\mathbf{t}}(n-1, \alpha + 1, \beta) + t_n S_{\mathbf{t}}(n - 1, \alpha, \beta + 1).
		\end{split}
	\end{equation*}
\end{proof}
\begin{remark}
	This identity has an interprepable combinatorial meaning in terms of ``\textit{weird balls experiment}`` (I coined the term, but there is a great chance that it exists in some obscure combinatorial literature under a different name). 
	
	There are $n$ \textcolor{blue}{blue balls} and $n + \alpha + \beta + 2$ \textcolor{red}{red balls}. All \textcolor{blue}{blue balls} are distinct, e.g. labeled. An $i$-th \textcolor{blue}{blue ball} has a weight $t_i$, whereas all \textcolor{red}{red balls} have a weight equal to $1$.
	There is also a line of scales with the capability of weighting balls. The experiment then goes as follows:
	
	\begin{enumerate}
		\item We select a random subset of \textcolor{blue}{blue balls}, let's denote the  selected \textcolor{blue}{blue balls} with indicator vector of size $n$ as $\mathbf{k}$;
		\item We take $\beta + 1$ \textcolor{red}{red balls} and label them arbitrary;
		\item We group selected \textcolor{blue}{blue} and \textcolor{red}{red} balls together and then put them on scales consequently (one after another, with no gaps, starting at the first position) in an arbibtrary order;
		\item We sample $n - |\mathbf{k}| + \alpha + 1$ \textcolor{red}{red balls} and label each one arbitrarily;
		\item We put them on scales consequently (starting with the first available position);
		\item We write down scales' measurements and compute their product $\xi = \mathbf{t}^{\mathbf{k}}$ (note that \textcolor{red}{red balls} do not contribute to the product as their weight is $1$).
	\end{enumerate} 
	
	Then, $S_{\mathbf{t}}(n, \alpha, \beta)$ is equal to the sum of all possible $\xi$ values (remember that $\Gamma(x) = (x - 1)!$ for integers, which is a number of permutations). Moreover, the above proposition starts to make a combinatorial sense: indeed, $S_{\mathbf{t}}(n, \alpha, \beta)$ is a sum of all cases where we did not select $n$-th \textcolor{blue}{blue ball} plus those where we did. The former is equal to $S_{\mathbf{t}}(n - 1, \alpha + 1, \beta)$ (remember that we take a \textcolor{red}{red ball} for each non-selected \textcolor{blue}{blue ball}) and the later is $t_n S_{\mathbf{t}}(n - 1, \alpha, \beta + 1)$ (we add an extra mandatory \textcolor{red}{red ball} to the pile of \textcolor{blue}{blue} and \textcolor{red}{red} balls, while correcting the $\xi$ value for $t_n$, as the new \textcolor{red}{red ball} as weight of $1$).
\end{remark}
This allows us to compute $S_{\mathbf{t}}(n, \alpha, \beta)$ with $O(n^{\color{orange}{3}})$ complexity via the dynamic programming. It is notoriously better than $O(\textcolor{red}{2^n})$, but we can do better.

\begin{proposition}
	The following relation holds:
	\begin{equation}\label{eq:marg_ab}
		S_{\mathbf{t}}(n, \alpha + 1, \beta) = (n + \alpha + \beta) S_{\mathbf{t}}(n, \alpha, \beta) - S_{\mathbf{t}}(n, \alpha, \beta + 1)
	\end{equation}

\end{proposition}
\begin{proof}
	\begin{equation*}
		\begin{split}
		S_{\mathbf{t}}(n, \alpha + 1, \beta) = \sum_{\mathclap{\mathbf{0}_n \le \mathbf{k} \le \mathbf{1}_n}} \mathbf{t}^\mathbf{k} \Gamma(n - |\mathbf{k}| + \alpha + 1) \Gamma(|\mathbf{k}| + \beta) = \sum_{\mathclap{\mathbf{0}_n \le \mathbf{k} \le \mathbf{1}_n}} \mathbf{t}^\mathbf{k} \Gamma(n - |\mathbf{k}| + \alpha) (n - |\mathbf{k}| + \alpha) \Gamma(|\mathbf{k}| + \beta) = \\
		= (n + \alpha)\sum_{\mathclap{\mathbf{0}_n \le \mathbf{k} \le \mathbf{1}_n}} \mathbf{t}^\mathbf{k} \Gamma(n - |\mathbf{k}| + \alpha) \Gamma(|\mathbf{k}| + \beta) - \sum_{\mathclap{\mathbf{0}_n \le \mathbf{k} \le \mathbf{1}_n}} \mathbf{t}^\mathbf{k} |\mathbf{k}| \Gamma(n - |\mathbf{k}| + \alpha) \Gamma(|\mathbf{k}| + \beta) = \\ = (n + \alpha) S_{\mathbf{t}}(n, \alpha, \beta) - \sum_{\mathclap{\mathbf{0}_n \le \mathbf{k} \le \mathbf{1}_n}} \mathbf{t}^\mathbf{k} (|\mathbf{k}| + \beta - \beta) \Gamma(n - |\mathbf{k}| + \alpha) \Gamma(|\mathbf{k}| + \beta) = (n + \alpha + \beta) S_{\mathbf{t}}(n, \alpha, \beta) - \\ - \sum_{\mathclap{\mathbf{0}_n \le \mathbf{k} \le \mathbf{1}_n}} \mathbf{t}^\mathbf{k}  \Gamma(n - |\mathbf{k}| + \alpha) \Gamma(|\mathbf{k}| + \beta + 1) = (n + \alpha + \beta) S_{\mathbf{t}}(n, \alpha, \beta) - S_{\mathbf{t}}(n, \alpha, \beta + 1)
		\end{split}
	\end{equation*}
\end{proof}
\begin{remark}
	Again, it makes combinatorial sense under the vein of the ``weird balls experiment``. It suffices to see that $S_{\mathbf{t}}(n, \alpha + 1, \beta) + S_{\mathbf{t}}(n, \alpha, \beta + 1) = (n + \alpha + \beta) S_{\mathbf{t}}(n, \alpha, \beta) $.
\end{remark}

Combining \autoref{eq:marg_rec1} and \autoref{eq:marg_ab}, we obtain the following formula:

\begin{equation}
	S_{\mathbf{t}}(n, \alpha, \beta) = (n + \alpha + \beta - 1)S_{\mathbf{t}}(n - 1, \alpha, \beta) + (t_n - 1)S_{\mathbf{t}}(n - 1, \alpha, \beta + 1),
\end{equation}
which is $O(n^{\color{green}{2}})$. It appears that we can't improve upon this asymptotics, but we can reduce the number of operations two-fold.

\begin{lemma}
	Let $P_{n}(x)$ be a polynomial with roots at integer values $[0, \dots, n - 1]$, known as a \textit{falling factorial}:
	$$P_{n}(x) = x_{(n)} = \prod_{i=0}^{n - 1} (x - i) = \sum_{i=0}^{n} (-1)^{n - i} \stirling{n}{i} x^{i},$$
	where $\stirling{n}{i}$ is the Stirling number of the first kind. Then, the coefficients of a same polynomial with roots translated by $h$, $T_{n}(x) = P_{n}(x - h) = \sum_{i=0}^{n} a_i x^i$, attain the form:
	$$a_m =  (-1)^{n - m}\sum_{k=m}^n  \stirling{n}{k} \binom{k}{m} h^{k-m}.$$
\end{lemma}
\begin{proof}
	By the binomial theorem,
	\begin{equation*}
		\begin{split}
			T_{n}(x) = \sum_{i=0}^{n} (-1)^{n - i} \stirling{n}{i}  (x - h)^{i} = \sum_{i=0}^{n} (-1)^{n - i} \stirling{n}{i}  \sum_{k=0}^{i} \binom{i}{k} x^k (-h)^{i - k}
		\end{split}
	\end{equation*}
	Then we collect terms in front of $x^m$:
	$$a_m = \sum_{k=0}^n (-1)^{n - k} \stirling{n}{k} \binom{k}{m} (-h)^{k-m} =  \sum_{k=m}^n (-1)^{n - k} \stirling{n}{k} \binom{k}{m} (-h)^{k-m} = (-1)^{n - m}\sum_{k=m}^n  \stirling{n}{k} \binom{k}{m} h^{k-m}.$$
%
\end{proof}
\begin{remark}
	The coefficient sum is a convolution of the sequence $\stirling{n}{k}\binom{k}{m}$ with the geometric sequence $h^{k}$. Applying the generating-function machinery to this convolution does not yield a form simpler than the one above, which we take to be final.
\end{remark}

\begin{proposition}\label{thr:existence}
	For any fixed $n$, $\alpha$ and $\beta$, $\{S_{\mathbf{t}}(n, \alpha, \beta + d)\}_{d=0}^{n + 1}$ are linearly dependent, i.e. there exist such $\{c_i\}_{i=0}^{n + 1}$ that
	$$\sum_{d=0}^{n + 1}c_d S_{\mathbf{t}}(n, \alpha, \beta + d) = 0.$$
\end{proposition}
\begin{proof}
	We start with some motivation first. Let's see decompositions of $S_{\mathbf{t}}(n, \alpha, \beta + 1), S_{\mathbf{t}}(n, \alpha, \beta + 2)$ and $S_{\mathbf{t}}(n, \alpha, \beta + d)$:
	\begin{equation*}
		\begin{split}
	S_{\mathbf{t}}(n, \alpha, \beta + 1) = \beta S_{\mathbf{t}}(n, \alpha, \beta) +  \sum_{\mathclap{\mathbf{0}_{n} < \mathbf{k} < \mathbf{1}_{n}}} \mathbf{t}^\mathbf{k}  \Gamma(n - |\mathbf{k}| + \alpha ) \Gamma(|\mathbf{k}| + \beta) |\mathbf{k}|\\
	S_{\mathbf{t}}(n, \alpha, \beta + 2) = \beta^2 S_{\mathbf{t}}(n, \alpha, \beta) + S_{\mathbf{t}}(n, \alpha, \beta + 1) + \sum_{\mathclap{\mathbf{0}_{n} < \mathbf{k} < \mathbf{1}_{n}}} \mathbf{t}^\mathbf{k}  \Gamma(n - |\mathbf{k}| + \alpha ) \Gamma(|\mathbf{k}| + \beta) (2 \beta |\mathbf{k}| +  |\mathbf{k}|^2) \\
	S_{\mathbf{t}}(n, \alpha, \beta + d) = \text{something} + \sum_{\mathclap{\mathbf{0}_{n} < \mathbf{k} < \mathbf{1}_{n}}} \mathbf{t}^\mathbf{k}  \Gamma(n - |\mathbf{k}| + \alpha ) \Gamma(|\mathbf{k}| + \beta) |\mathbf{k}|^d.
		\end{split}
	\end{equation*}
	Notice how computation of the $S_{\mathbf{t}}(n, \alpha, \beta + d)$ requires computing a sum with terms multiplied by $|\mathbf{k}|^d$. The value of this sum is unknown without the knowledge of previous terms. However, $|\mathbf{k}|$ attains integer values in the interval $[0, n]$, therefore we can construct the polynomial of degree $n + 1$ that attains zero at those values, possibly shifted by $\beta$. Trivially, this polynomial is the one from the lemma on the shifted \textit{falling factorial} above with $h=\beta$ and coeffients $a_m = (-1)^{n + 1 - m}\sum_{k=m}^{n+1}  \stirling{n + 1}{k} \binom{k}{m} \beta^{k-m}$. The \textit{rising factorial}, as a matter of fact, can be expressed similarily: $\prod_{i=0}^n (x + i) = x^{(n + 1)} = \sum_{m=0}^{n + 1} \stirling{n + 1}{m} x^m$. Those observations will come in handy soon enough. 
	
	We are looking for a set of a coefficients $\{c_i\}_{i=1}^{n + 1}$ for which the following equation holds:
	$$\sum_{d=0}^{n + 1} c_{d} S_{\mathbf{t}}(n, \alpha, \beta + d) = 0.$$
	First, see that $S_{\mathbf{t}}(n, \alpha, \beta + d) = \sum_{{\mathbf{0}_{n} < \mathbf{k} < \mathbf{1}_{n}}} \mathbf{t}^\mathbf{k}  \Gamma(n - |\mathbf{k}| + \alpha ) \Gamma(|\mathbf{k}| + \beta) (|\mathbf{k}| + \beta)^{(d)}$. For brevity, we shall denote $\mathbf{t}^\mathbf{k}  \Gamma(n - |\mathbf{k}| + \alpha ) \Gamma(|\mathbf{k}| + \beta)$ as $z(\mathbf{k})$. Then,
	\begin{equation*}
		\begin{split}
			\sum_{d=0}^{n + 1} c_{d} S_{\mathbf{t}}(n, \alpha, \beta + d) = \sum_{d=0}^{n + 1} c_d \sum_{\mathclap{\mathbf{0}_{n} < \mathbf{k} < \mathbf{1}_{n}}} z(\mathbf{k})  (|\mathbf{k}| + \beta)^{(d)} =  \sum_{\mathclap{\mathbf{0}_{n} < \mathbf{k} < \mathbf{1}_{n}}} z(\mathbf{k}) \sum_{d=0}^{n + 1} c_d (|\mathbf{k}| + \beta)^{(d)} = \\ = \sum_{\mathclap{\mathbf{0}_{n} < \mathbf{k} < \mathbf{1}_{n}}} z(\mathbf{k}) \sum_{d=0}^{n + 1} c_d \sum_{m=0}^{d} \stirling{d}{m} (|\mathbf{k}| + \beta)^m = 0 \leftrightarrow \sum_{d=0}^{n + 1} c_d \sum_{m=0}^{d} \stirling{d}{m} (|\mathbf{k}| + \beta)^m = \sum_{0 \le m  \le d}^{d \le n + 1} c_d \stirling{d}{m} (|\mathbf{k}| + \beta)^m = 0.
		\end{split}
	\end{equation*}
	For this to work, it suffices to find $c_d$ such that coefficients in front of $(|\mathbf{k}| + \beta)$ powers agree with corresponding series coefficients of the falling factorial $(\mathbf{k} + \beta)_{(n + 1)} = \sum_{m=0}^{n + 1} \stirling{n + 1}{m} (-1)^{n - m + 1}(|\mathbf{k}| + \beta)^m$:\\
	\begin{center}
	\begin{tabular}{ll}
		$m = 0$ & $c_0 = (-1)^{n + 1}\beta^{(n + 1)}$ \\
		$m = 1$ & $\sum_{d=1}^{n + 1} c_d\stirling{d}{1} = (-1)^{n }\sum_{k=1}^{n + 1}  \stirling{n + 1}{k} k \beta^{k-1}$\\
		$\vdots$ & \\
		$m = j$ & $\sum_{d=j}^{n + 1} c_d\stirling{d}{j} = (-1)^{n - j + 1}\sum_{k=j}^{n + 1}  \stirling{n + 1}{k} \binom{k}{j} \beta^{k-j}$\\
		$\vdots$& \\
		$m = n + 1$ & $c_{n + 1} = 1$
	\end{tabular}
   \end{center}
   This constitutes a triangular system of $n + 1$ equations $A \mathbf{c} = \mathbf{s}$, where $A_{i,j} = \stirling{j}{i}$ for all $j \ge i$ and $0$ otherwise, $\mathbf{c}_{i} = c_i$, $s_{j} = (-1)^{n - j + 1}\sum_{k=j}^{n + 1}  \stirling{n + 1}{k} \binom{k}{j} \beta^{k-j}$, $i, j \in [1,\dots,n + 1]$. The solution to the system exists as $det(A) = \prod_{i=0}^{n + 1} \stirling{i}{i} = 1$. The values of $\mathbf{c}$ can be obtained with backward substitution.
\end{proof}
\begin{remark}
	There is probably a combinatorial explanation for this nonsense too... I suppose. High IQ individuals are free to express their very noble opinion.
\end{remark}

The result obtained above is nice, but practically of little use as we still have to solve a system of linear equations in $O((n + 1)^2)$ time. However, it turns out that we can find constant coefficients (independent of $\beta$) if we solve a similar, easily transformable to the original, problem. But to proceed further, we first need one well-known to those acquainted to combinatorics and the Stirling numbers lemma:

\begin{lemma}
	Let $A$ and $B$ be matrices of size $n$ such that:
	$$A_{i,j} = \stirling{i}{j},~\forall j \le i,~\text{else } 0,$$
	$$B_{i,j} = (-1)^{i - j}\stirlings{i}{j},~\forall j \le i,~\text{else } 0,$$
	where $\stirlings{i}{j}$ is the Stirling number of the second kind that can be defined as $x^i = \sum_{j=0}^i \stirlings{i}{j}  (-1)^{i - j} x^{(j)}$. Then, 
	$$A B = I,$$
	i.e. $B = A^{-1}$.
\end{lemma}
\begin{proof}
	It suffices to see that $A$ is a change of basis matrix from polynomials $\{x^i\}$ to rising factorials and $B$ is the opposite.
\end{proof}
\begin{corollary}
	Given the system $A \mathbf{c} = \mathbf{s}$, where $A_{i,j} = \stirling{j}{i} ~\forall j \ge i$ and $0$ otherwise (an upper triangular matrix of Stirling values of the first kind), $\mathbf{c}_{i} = c_i$, $s_{j} = (-1)^{n - j + 1}\sum_{k=j}^{n + 1}  \stirling{n + 1}{k} \binom{k}{j} \beta^{k-j}$, $i, j \in [1,\dots,n + 1]$, the solution is $\mathbf{c} = A^{-1} \mathbf{s}$. This $A$ is the transpose of the matrix in the Lemma, so $(A^\top)^{-1} = (A^{-1})^\top$ is the upper triangular matrix of (signed) Stirling numbers of the second kind:
	$$\mathbf{c}_i = \sum_{j=i}^{n + 1} (-1)^{i - j}\stirlings{j}{i} s_j = \sum_{j=i}^{n + 1} (-1)^{i + n + 1}\stirlings{j}{i}  \sum_{k=j}^{n + 1}  \stirling{n + 1}{k} \binom{k}{j} \beta^{k-j}.$$
\end{corollary}

Also, the following lemma will be needed very soon. I call it a lemma rather than a conjecture as it is kinda proven, but I am still not quiet happy that I didn't provide a strong algebraic proof.
\begin{figure}
	\includegraphics[height=\textheight]{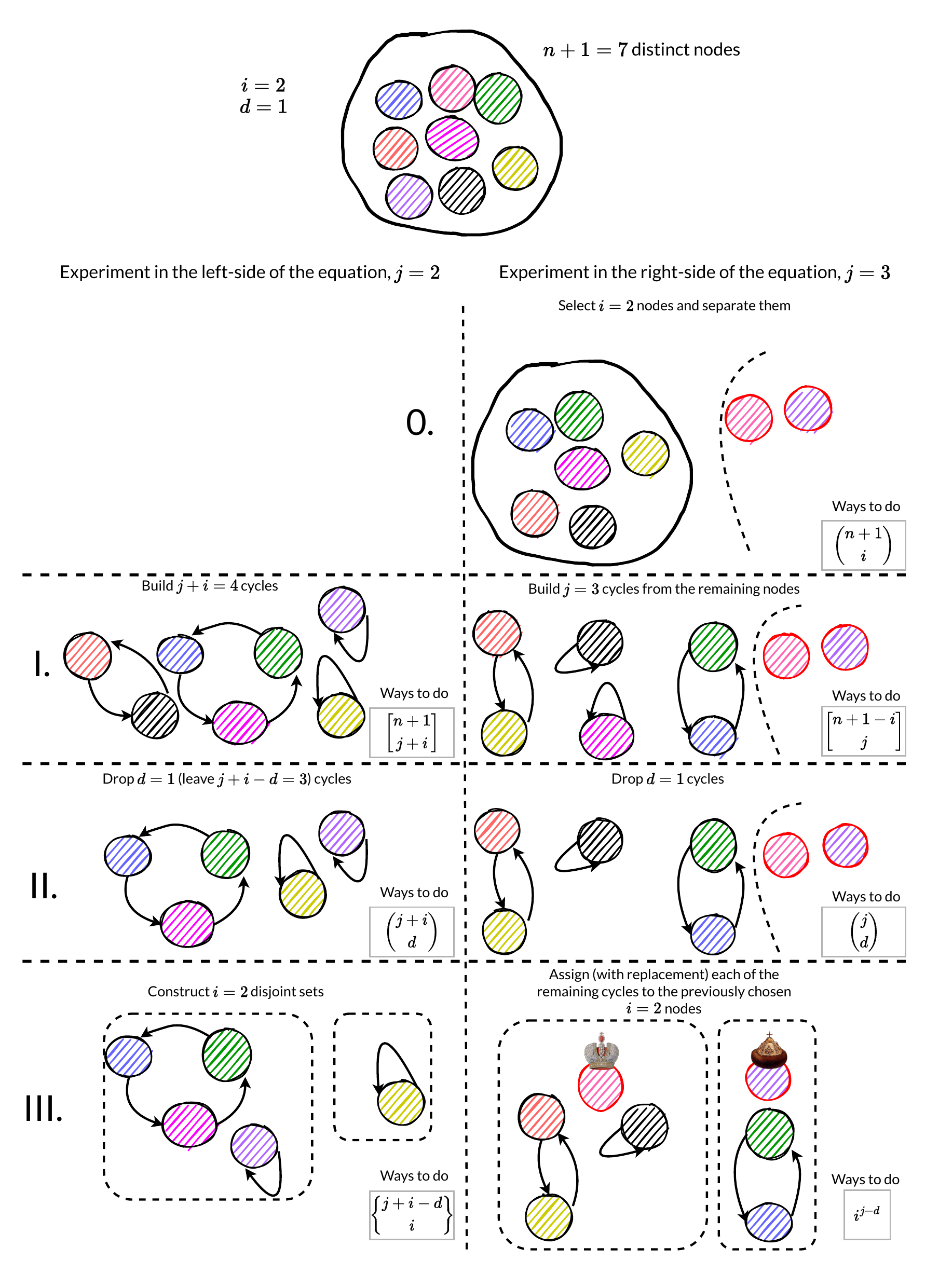}
	\caption{Schematic representation of the experiments ``run`` in the left and right sides of the equation.}
	\label{fig:conjecture}
\end{figure}
\begin{conjecture}\label{lem:coronation}
	$$ \sum_{j=d}^{n + 1 - i}  \stirlings{j + i - d}{i} \stirling{n + 1}{j + i} \binom{j + i}{d} = \binom{n +1}{i} \sum_{j=d}^{n + 1 - i} \stirling{n + 1 - i}{j} \binom{j}{d} i^{j - d}~~\forall n, d, i\in \mathbb{Z}_{+}.$$
\end{conjecture}
\begin{proof}
	We can think of this relation in a combinatorial way. Remember interpretation of the following combinatorial values:
	\begin{center}
		\renewcommand{\arraystretch}{1.5}
	\begin{tabular}{lll}
		
		Binomial coefficient& $\binom{n}{k}$ & Number of ways to choice $k$ items out of $n$ distinct elements \\
		Stirling number of the 1st kind &$\stirling{n}{k}$ & Number of ways to construct $k$ disjoint cycles from $n$ nodes \\
		Stirling number of the 2st kind & $\stirlings{n}{k}$ & Number of ways to separate $n$ elemenets into $k$ disjoint sets \\
		Integer power & $n^k$ & Number of ways to choose $k$ items from a set of size $n$ with replacement
	\end{tabular}
	\renewcommand{\arraystretch}{1}
	\end{center}
	
	In the left side of the equation, $\sum_{j=d}^{n + 1 - i}  \stirlings{j + i - d}{i} \stirling{n + 1}{j + i} \binom{j + i}{d} = \sum_{j=d}^{n + 1 - i} b_j$, $b_j$ is a numer of ways the following experiment can be conducted:
	\begin{enumerate}
		\item Constructing $j + i$ cycles from the pool of $n + 1$ elements (can be done in $\stirling{n + 1}{j + i}$ number of ways);
		\item Dropping $d$ of the constructed cycles (can be done in $\binom{j + i}{d}$ number of ways);
		\item Grouping the remaining cycles into $i$ disjoint groups. 
	\end{enumerate}
	The $d_k$ in the right side of the equation, $\binom{n +1}{i} \sum_{j=d}^{n + 1 - i} \stirling{n + 1 - i}{j} \binom{j}{d} i^{j - d} = \sum_{j=d}^{n + 1 - i} d_k$, corresponds to a number of ways to perform a similar, but different experiment:
	\begin{enumerate}
		\setcounter{enumi}{-1}
		\item Choosing $i$ elements from the pool of $n + 1$ elements and separating them;
		\item Constructing $j$ cycles from the remaining $n + 1 - i$ elements (can be done in $\stirlings{n + 1 - i}{j}$ number of ways);
		\item Dropping $d$ of the constructed cycles (can be done in $\binom{j}{d}$ number of ways);
		\item Assigning the remaining cycles to some of the items selected at the first step (can be done in $i^{j-d}$ number of ways).
	\end{enumerate}
	The schematic representation of experiments is depicted in \autoref{fig:conjecture}. The most striking difference is that at the last step, in the left-side experiment, cycles are separated into  $i$ disjoint non-empty sets, whereas in the right-side experiment, the last step can be interpreted as allowing some of the $i$ sets to be empty. However, we can take a look at it from a different angle and see that all of the $i$ ``sets`` already have an element in them -- one of the selected $i$ items. Then, experiments look more similar. The difference is, if we start looking at the right-side experiment, it can be interpreted as being conditioned on the assumption that each of $i$ sets has at least $1$ self-cycle. If we can find a bijective transformation of an outcome from the left-side experiment to the right-side experiment, then we can conclude outcome spaces of experiments have the same size, therefore the combinatorial identity in question should be correct.
	
	Next, we shall call the $i$ selected balls as \textit{tsarballs} and other balls will be called \textit{subballs}. We need to find a \textbf{coronation} process to select a tsarball from a set of subballs, and vice-versa, the process of \textbf{deposition} of a tsarball, that converts a tsarball into a subball. Furthermore, the coronotation and deposition processes should be inverses of one another: we should be able to run a deposition after the coronation and get the initial set. To this end, we impose a total ordering on all balls -- a completely arbitrary one (e.g. we can just provide numbers from $1$ to $n$ to each ball) -- it will cause no change to the left-side and the right-side experiment, but will play a role in the coronation and deposition processes.

	\textbf{Coronation:}
	\begin{enumerate}
		\item Find the ''largest`` (in the total ordering sense) ball in a set;
		\item Select a ball it points to as a tsarball (note that if a ball is in a self-loop, it points to itself).
	\end{enumerate}
	
	\textbf{Deposition:}
	\begin{enumerate}
		\item Find the largest ball in a set;
		\item If the largest ball in a set is the current tsarball, make the tsarball a subball in a self-loop, otherwise continue to the next step;
		\item Depose the current tsarball and insert it into the largest ball's loop right after the largest ball (so the largest ball points to the ex-tsarball).
	\end{enumerate}
	
	Those processes are indeed bijective and are an inverse of one another. We can apply the deposition process to an outcome and each subset of the right-side experiment and obtain an outcome of the left-side experiment, and vice-versa, we can perfrom a coronation process on the left-side experiment to obtain the right-side experiment.
\end{proof}

\begin{proposition}
	$$c_i = \sum_{j=i}^{n + 1} (-1)^{i + n + 1}\stirlings{j}{i}  \sum_{k=j}^{n + 1}  \stirling{n + 1}{k} \binom{k}{j} \beta^{k-j} = (-1)^{n + 1 - i}(\beta + i)^{(n + 1 - i)} \binom{n + 1}{i},$$
	as long as \autoref{lem:coronation} holds.
\end{proposition}
\begin{proof}
	Firstly, let's provide a power series expansion for the $(\beta + i)^{(n + 1 - i)} \binom{n + 1}{i}$ with resepct to $\beta$:
	
	\begin{equation*}
	\begin{split}
		(\beta + i)^{(n + 1 - i)} \binom{n + 1}{i} = \binom{n +1}{i} \sum_{j=0}^{n + 1 - i} \stirling{n + 1 - i}{j} (\beta + i)^j =\binom{n +1}{i} \sum_{j=0}^{n + 1 - i} \stirling{n + 1 - i}{j} \sum_{k=0}^{j} \binom{j}{k} \beta^k i^{j - k} = \\ = \binom{n +1}{i}\sum_{d=0}^{n + 1 - i} \beta^d \left( \sum_{j=d}^{n + 1 - i} \stirling{n + 1 - i}{j} \binom{j}{d} i^{j - d}\right)
	\end{split}
	\end{equation*}
	
	Secondly, we collect terms in front of $\beta$ for the $\sum_{j=i}^{n + 1} \stirlings{j}{i}  \sum_{k=j}^{n + 1}  \stirling{n + 1}{k} \binom{k}{j} \beta^{k-j}$:
	
	\begin{equation*}
		\begin{split}
			\sum_{j=i}^{n + 1} (-1)^{i + n + 1}\stirlings{j}{i}  \sum_{k=j}^{n + 1}  \stirling{n + 1}{k} \binom{k}{j} \beta^{k-j} = \sum_{d=0}^{n + 1 - i}\beta^d \left(\sum_{j=i}^{n + 1 - d}  \stirlings{j}{i} \stirling{n + 1}{d + j} \binom{d + j}{j} \right) = \\ =
			\sum_{d=0}^{n + 1 - i}\beta^d \left(\sum_{j=d}^{n + 1 - i}  \stirlings{j + i - d}{i} \stirling{n + 1}{j + i} \binom{j + i}{d} \right).
		\end{split}
	\end{equation*}
	
Therefore, statement in the proposition is true if the folowing equality holds:
$$ \sum_{j=d}^{n + 1 - i}  \stirlings{j + i - d}{i} \stirling{n + 1}{j + i} \binom{j + i}{d} = \binom{n +1}{i} \sum_{j=d}^{n + 1 - i} \stirling{n + 1 - i}{j} \binom{j}{d} i^{j - d}.$$
It appears so, as stated by \autoref{lem:coronation}.
\end{proof}

\begin{figure}
	\centering
	\includegraphics[width=0.8\textwidth]{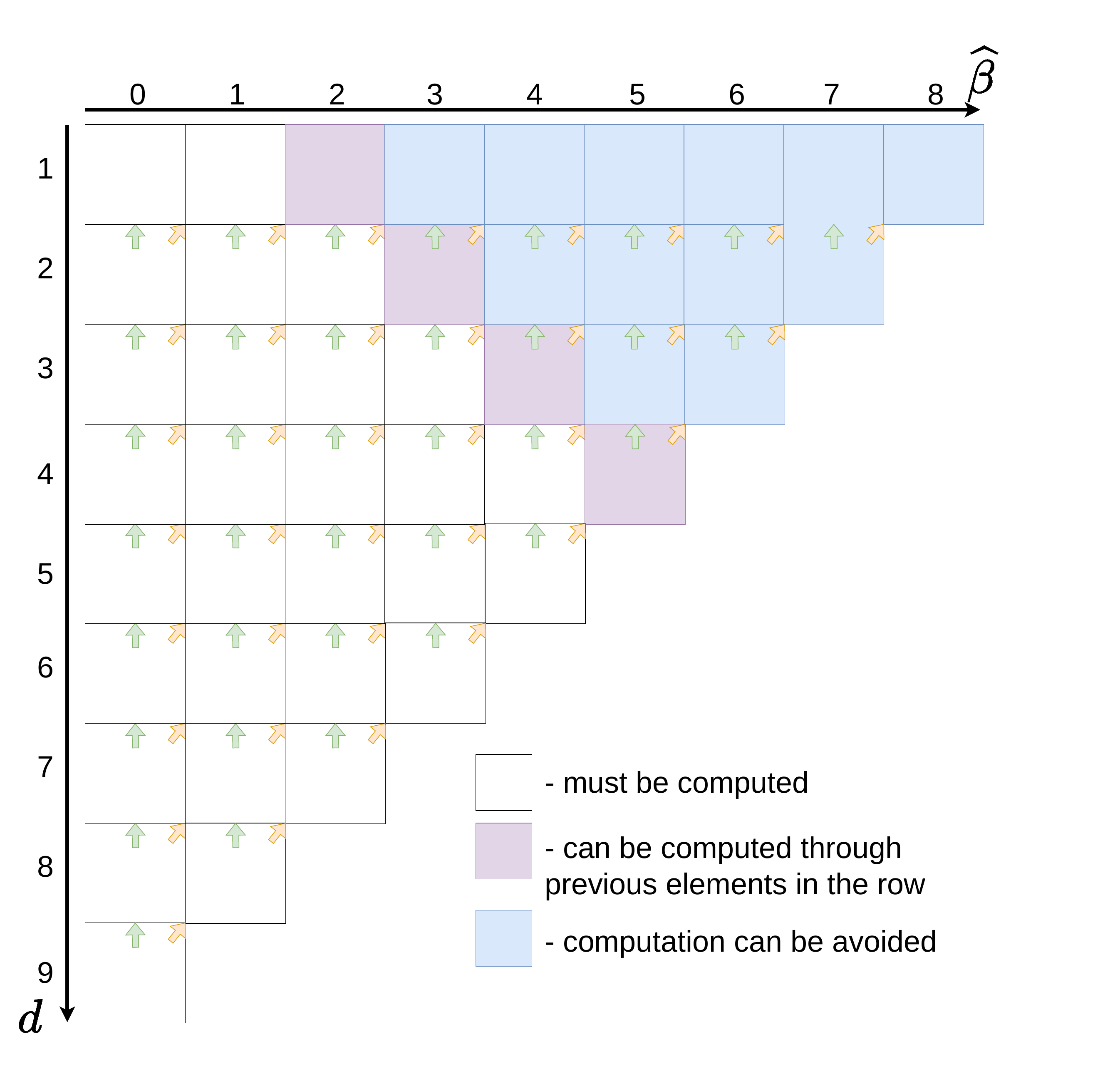}
	\caption{The dynamic programming matrix. Excessive computations are avoided with the help of the Proposition~\autoref{thr:existence}. $\widehat{\beta}$ is a value of positive integer shift to the $\beta$ variable.}
	\label{fig:dp}
\end{figure}

The implications of the just-proven Proposition is that we only need a quarter of the DP matrix (see \autoref{fig:dp}). Also, in scenarios where $\alpha$ or $\beta$ are integers (or only integer shifts are of interest), it provides as efficient way to compute $S$ with respect to integer shifts. That can greatly fasciliate optimization when done in $\mathbb{Z}$-space.

\subsection{Exact posterior of the mixture weight}\label{sec:posterior}

The marginal likelihood is not only the evidence used for empirical-Bayes fitting and model comparison; it also yields the full posterior of $w$ in closed form. Writing $t_i = f(x_i\mid\theta)/g(x_i\mid\theta)$ (the reciprocal of the ratio used in \autoref{eq:polynomial}, chosen here so that the weighted component is $f$) and letting $e_m(\mathbf{t})$ denote the elementary symmetric polynomials, the marginal admits the equivalent form
\begin{equation}\label{eq:marg_esp}
	\int_0^1 \widetilde{L}(\mathpzc{X}\mid\theta, w)\, q(w\mid\alpha,\beta)\, dw = \left(\prod_{i=1}^n g(x_i\mid\theta)\right) \sum_{m=0}^n e_m(\mathbf{t})\, \frac{B(m+\alpha,\, n-m+\beta)}{B(\alpha,\beta)},
\end{equation}
obtained by writing $\prod_i\big(w f(x_i) + (1-w) g(x_i)\big) = \big(\prod_i g(x_i)\big) \prod_i\big(w\, t_i + (1-w)\big) = \big(\prod_i g(x_i)\big) \sum_{m=0}^n e_m(\mathbf{t})\, w^m (1-w)^{n-m}$ and integrating each term against the Beta prior (this is \autoref{eq:polynomial} rearranged, with all summands positive, so it is evaluated stably in the log-domain). By Bayes' theorem the posterior is therefore a \emph{finite mixture of Beta distributions},
\begin{equation}\label{eq:post_mixture}
	p(w \mid \mathpzc{X}) = \sum_{m=0}^n \pi_m\, \mathpzc{Beta}(w \mid m+\alpha,\; n-m+\beta), \qquad \pi_m = \frac{e_m(\mathbf{t})\, B(m+\alpha,\, n-m+\beta)}{\sum_{m'=0}^n e_{m'}(\mathbf{t})\, B(m'+\alpha,\, n-m'+\beta)}.
\end{equation}
Every posterior summary is then available in closed form: the mean $\EX[w\mid\mathpzc{X}] = \sum_m \pi_m \frac{m+\alpha}{n+\alpha+\beta}$, the variance by the law of total variance, and the CDF (hence quantiles and equal-tailed credible intervals) as the $\pi_m$-weighted sum of regularised incomplete Beta functions.

The posterior class probability of an individual observation follows from the same object. Since, conditional on $w$, the latent label $z_j$ is independent of the remaining data,
\begin{equation}\label{eq:responsibility}
	P(z_j = f \mid \mathpzc{X}) = \EX_{w\mid\mathpzc{X}}\!\left[\frac{w\, f(x_j)}{w\, f(x_j) + (1-w)\, g(x_j)}\right] = t_j \frac{\sum_{m=0}^{n-1} e_m(\mathbf{t}_{-j})\, B(m+\alpha+1,\, n-1-m+\beta)}{\sum_{m=0}^{n} e_m(\mathbf{t})\, B(m+\alpha,\, n-m+\beta)},
\end{equation}
where $\mathbf{t}_{-j}$ omits observation $j$; the leave-one-out elementary symmetric polynomials follow from the full ones by deflation, so the identity gives all $n$ probabilities in $O(n^2)$ (for a wide likelihood-ratio range our reference implementation instead recomputes them in the log domain, trading this bound for numerical stability). In the two-groups model $g$ is the null density and $1 - P(z_j = f\mid\mathpzc{X})$ is exactly the \emph{local false discovery rate} of observation $j$. Thus a single $O(n^2)$ computation returns the marginal likelihood, the entire posterior of the mixing proportion with calibrated uncertainty, and the local FDR of every observation -- with no sampling and no EM iterations.

\section{Structural properties of the exact posterior}\label{sec:properties}

The mixture representation \eqref{eq:post_mixture} is not merely a means of computation; it exposes structure that the exact posterior is \emph{guaranteed} to possess and that no finite-sample approximation is guaranteed to share. Three such properties -- log-concavity, an evidence-ratio identity for the moments, and monotonicity in the data -- are recorded here. All three descend from a single classical fact about the elementary symmetric polynomials, which makes the two-groups posterior an unusually well-behaved object among finite mixtures.

Throughout, $t_i = f(x_i\mid\theta)/g(x_i\mid\theta) > 0$ are the per-observation likelihood ratios, $e_m(\mathbf{t})$ their elementary symmetric polynomials, and
\begin{equation}\label{eq:count_weights}
	\pi_m = \frac{e_m(\mathbf{t})\, B(m+\alpha,\, n-m+\beta)}{\sum_{m'=0}^n e_{m'}(\mathbf{t})\, B(m'+\alpha,\, n-m'+\beta)}, \qquad m = 0,\dots,n,
\end{equation}
the mixing weights of \eqref{eq:post_mixture}. Conditional on $w$, the number of observations drawn from the signal component $f$ is $N_f := \#\{i : z_i = f\}$; a direct computation shows that $\{\pi_m\}$ is exactly the posterior law of this latent count, $\pi_m = P(N_f = m \mid \mathpzc{X})$. Every statement below therefore concerns a quantity of direct scientific interest -- how many of the $n$ cases are non-null -- and not only the abstract mixing weight.

The generating polynomial of the likelihood ratios,
\begin{equation}\label{eq:genpoly}
	\prod_{i=1}^n (1 + t_i z) = \sum_{m=0}^n e_m(\mathbf{t})\, z^m,
\end{equation}
has all its roots at $z = -1/t_i$, which are real and, since $t_i > 0$, negative. Real-rootedness with nonnegative coefficients is precisely the condition under which a coefficient sequence is a \emph{P\'olya frequency sequence}; it forces Newton's inequalities.

\begin{lemma}[Newton's inequalities]\label{lem:newton}
	For every $\mathbf{t} \in \mathbb{R}_{>0}^n$ the normalised elementary symmetric polynomials $E_m := e_m(\mathbf{t})/\binom{n}{m}$ satisfy
	\begin{equation}\label{eq:newton}
		E_m^2 \ge E_{m-1}\, E_{m+1}, \qquad m = 1,\dots,n-1,
	\end{equation}
	with equality throughout iff all $t_i$ coincide. Consequently $\{e_m(\mathbf{t})\}_{m=0}^n$ is strictly log-concave and has no internal zeros.
\end{lemma}

Newton's inequalities \eqref{eq:newton} are the classical strengthening of the arithmetic--geometric-mean inequality \citep{hardy1952inequalities}; the short proof is given in \autoref{app:properties_proofs}. Their statistical consequence for the marginalised two-groups model is immediate.

\begin{theorem}[Log-concavity and unimodality of the prevalence-count posterior]\label{thm:logconcave}
	Let $w \sim \mathpzc{Beta}(\alpha,\beta)$ with $\alpha \ge 1$ and $\beta \ge 1$, and let $t_i > 0$ for all $i$. Then the posterior law $\{\pi_m\}_{m=0}^n$ of the signal count $N_f$ is log-concave,
	\begin{equation}
		\pi_m^2 \ge \pi_{m-1}\,\pi_{m+1}, \qquad m = 1,\dots,n-1,
	\end{equation}
	and is therefore unimodal with a contiguous set of modes. Under the uniform prior $\alpha=\beta=1$ the weights reduce to the normalised symmetric functions, $\pi_m \propto E_m$, and the statement is exactly Newton's inequality \eqref{eq:newton}.
\end{theorem}

\begin{proof}
	Write $\pi_m \propto E_m\, b_m$ with $b_m := \binom{n}{m} B(m+\alpha,\, n-m+\beta)$. The factor $E_m$ is log-concave by \autoref{lem:newton}. The factor $b_m$ is proportional to the Beta--Binomial$(n,\alpha,\beta)$ probability mass function, which is log-concave whenever $\alpha \ge 1$ and $\beta \ge 1$: its ratio
	\begin{equation}\label{eq:bb_ratio}
		\frac{b_{m+1}}{b_m} = \frac{(n-m)\,(m+\alpha)}{(m+1)\,(n-m-1+\beta)}
	\end{equation}
	is non-increasing in $m$ under these conditions (\autoref{app:properties_proofs}). A product of two log-concave sequences is log-concave, and a log-concave sequence with no internal zeros is unimodal with contiguous modes. For the uniform prior, $B(m+1,\,n-m+1) = \big[(n+1)\binom{n}{m}\big]^{-1}$, so $\pi_m \propto e_m(\mathbf{t})/\binom{n}{m} = E_m$.
\end{proof}

Both hypotheses are sharp. For a spiky prior with $\alpha < 1$ or $\beta < 1$ the Beta--Binomial factor is itself U-shaped and log-concavity of $\{\pi_m\}$ can fail; the ratio \eqref{eq:bb_ratio} ceases to be monotone. The requirement $\alpha,\beta \ge 1$ asks only that the prior on $w$ not concentrate at an endpoint, and is met by the uniform, the informative $\mathpzc{Beta}(1,4)$, and every weakly-informative prior used in \autoref{sec:numerical}.

\begin{remark}[A Poisson--Binomial view]\label{rem:poibin}
	Conditional on $w$ the labels $z_i$ are independent, so the signal count is a Poisson--Binomial variable, $N_f \mid w, \mathpzc{X} \sim \mathrm{PoiBin}\big(r_1(w),\dots,r_n(w)\big)$ with $r_i(w) = w t_i/(w t_i + 1 - w)$, whose generating polynomial $\prod_i\big(1 - r_i(w) + r_i(w)\, z\big)$ is real-rooted. The exact posterior $\{\pi_m\}$ is therefore a $\mathpzc{Beta}(\alpha,\beta)$-mixture of Poisson--Binomial laws. Log-concavity is \emph{not} automatically inherited under mixing -- averaging log-concave sequences can yield a multimodal one -- so \autoref{thm:logconcave} is a genuine statement about these particular Beta--Binomial mixing weights, and it situates the two-groups posterior within the theory of real-rooted (P\'olya frequency) sequences and their preservation properties.
\end{remark}

Log-concavity is exactly the property that makes the closed-form summaries of \autoref{sec:posterior} well-behaved as \emph{decisions}, not just as numbers.

\begin{corollary}[Single mode, coherent intervals]\label{cor:concentration}
	Under the hypotheses of \autoref{thm:logconcave}: (i) the posterior mode of $N_f$ is the largest $m$ with $e_m(\mathbf{t})/e_{m-1}(\mathbf{t}) \ge (n-m+\beta)/(m-1+\alpha)$, and once this fails at some $m$ it fails for all larger $m$, so the mode is located by a single $O(n)$ scan; (ii) $\{\pi_m\}$ has geometrically decaying tails away from the mode, so every posterior credible set for $N_f$ is a contiguous interval $\{m : a \le m \le b\}$, and the equal-tailed interval of \autoref{sec:posterior} is genuinely an interval rather than a union of disconnected pieces.
\end{corollary}

\autoref{fig:structural}(a) shows the shape these results guarantee on a rare-signal example: the count posterior is a single log-concave hump with a contiguous $95\%$ interval, where the EM estimate is a bare point.

\begin{figure}[t]
	\centering
	\includegraphics[width=\textwidth]{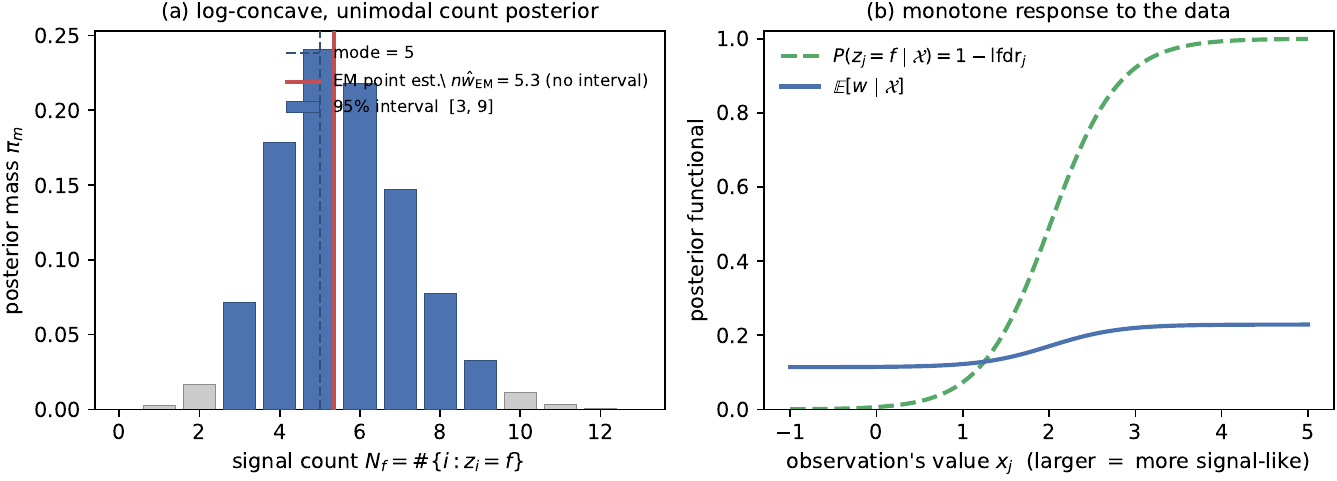}
	\caption{Structural properties of the exact two-groups posterior (Gaussian model $x_i \sim w\,\mathcal N(2.5,1) + (1-w)\,\mathcal N(0,1)$, uniform prior). \textbf{(a)} The posterior over the latent signal count $N_f$ ($n=24$) is log-concave and unimodal (\autoref{thm:logconcave}), with a single mode and a contiguous equal-tailed $95\%$ credible interval (\autoref{cor:concentration}); the EM point estimate supplies a number but no interval. \textbf{(b)} Monotone response to the data (\autoref{prop:monotone}, \autoref{cor:decision}): sweeping the value $x_j$ of one observation, both the posterior mean prevalence $\EX[w\mid\mathpzc X]$ and that observation's signal probability $P(z_j=f\mid\mathpzc X) = 1-\mathrm{lfdr}_j$ increase monotonically -- more signal-like data can only raise the inferred prevalence and lower the local FDR.}
	\label{fig:structural}
\end{figure}

The second structural fact is an exact identity that delivers all posterior moments of $w$ as by-products of the evidence computation, at shifted hyperparameters.

\begin{proposition}[Evidence-ratio moment identity]\label{prop:moments}
	Let $T(a,b) := \sum_{m=0}^n e_m(\mathbf{t})\, B(m+a,\, n-m+b)$, so that the marginal likelihood is $\big(\prod_i g(x_i)\big)\, T(\alpha,\beta)/B(\alpha,\beta)$. Then for every real $r > -\alpha$ the posterior raw moments of the mixing weight are ratios of the marginal at an $\alpha$-shifted hyperparameter,
	\begin{equation}\label{eq:moment_identity}
		\EX[w^r \mid \mathpzc{X}] = \frac{T(\alpha+r,\, \beta)}{T(\alpha,\, \beta)}.
	\end{equation}
	In particular the posterior mean and variance,
	\begin{equation}
		\EX[w \mid \mathpzc{X}] = \frac{T(\alpha+1,\beta)}{T(\alpha,\beta)}, \qquad
		\mathrm{Var}[w \mid \mathpzc{X}] = \frac{T(\alpha+2,\beta)}{T(\alpha,\beta)} - \left(\frac{T(\alpha+1,\beta)}{T(\alpha,\beta)}\right)^{\!2},
	\end{equation}
	require no summation over the mixture: each is a fresh evaluation of the same $O(n^2)$ marginal-likelihood routine with $\alpha$ incremented by an integer.
\end{proposition}

\begin{proof}
	With $g$ absorbed into $T$, the unnormalised posterior density of $w$ is proportional to $\prod_i\big(w f(x_i) + (1-w) g(x_i)\big)\, w^{\alpha-1}(1-w)^{\beta-1}$. Multiplying by $w^r$ raises the Beta exponent $\alpha \mapsto \alpha + r$; integrating term by term against \eqref{eq:marg_esp} replaces $B(m+\alpha, n-m+\beta)$ by $B(m+\alpha+r, n-m+\beta)$ in both numerator and normaliser, i.e.\ $T(\alpha,\beta) \mapsto T(\alpha+r,\beta)$. The factor $\prod_i g(x_i)$ cancels.
\end{proof}

The third fact certifies that the estimator responds to the data in the only direction that is scientifically coherent: strengthening the evidence for signal in any one observation can only raise the inferred prevalence.

\begin{proposition}[Monotone response to the data]\label{prop:monotone}
	Fix $\alpha, \beta$ and all likelihood ratios except the $i$-th. As a function of $t_i > 0$, the count posterior $\{\pi_m\}$ is monotone in the likelihood-ratio order: for $t_i < t_i'$ the ratio $\pi_m(t_i')/\pi_m(t_i)$ is non-decreasing in $m$. Consequently the posterior signal count $N_f$ is stochastically increasing in every $t_i$, and so are the posterior mean prevalence $\EX[w\mid\mathpzc{X}]$ and every upper-tail probability $P(N_f \ge k \mid \mathpzc{X})$.
\end{proposition}

\begin{proof}
	By multi-affinity of the elementary symmetric polynomials, $e_m(\mathbf{t}) = a_m + t_i\, a_{m-1}$, where $a_m := e_m(\mathbf{t}_{-i}) \ge 0$ are the leave-one-out symmetric functions, themselves log-concave by \autoref{lem:newton}. For $t_i < t_i'$ and $m < m'$, the Beta--Binomial factor and the normalisers cancel from the cross-difference, leaving
	\begin{equation*}
		(a_{m'} + t_i' a_{m'-1})(a_m + t_i a_{m-1}) - (a_m + t_i' a_{m-1})(a_{m'} + t_i a_{m'-1}) = (t_i' - t_i)\big(a_m a_{m'-1} - a_{m-1} a_{m'}\big) \ge 0,
	\end{equation*}
	the final inequality being the log-concavity of $\{a_m\}$ (Karlin's criterion: $a_m a_{m'-1} \ge a_{m-1} a_{m'}$ for $m \le m'-1$). This is the total-positivity ($\mathrm{TP}_2$) condition defining the monotone-likelihood-ratio order; stochastic monotonicity of $N_f$, and hence of every non-decreasing functional of it, follows.
\end{proof}

Monotonicity has a direct consequence for the multiple-testing use of the model (\autoref{sec:app_meta}--\ref{sec:app_genomics}), where the per-observation posterior \eqref{eq:responsibility} is thresholded to make discoveries.

\begin{corollary}[Coherent testing rule]\label{cor:decision}
	The local false discovery rate $\mathrm{lfdr}_j := 1 - P(z_j = f \mid \mathpzc{X})$ is strictly decreasing in the likelihood ratio $t_j$, and $P(z_j = f \mid \mathpzc{X})$ is non-decreasing in every other $t_i$. Hence the Bayes rule that rejects $\{\,j : \mathrm{lfdr}_j \le c\,\}$ is a threshold rule on the likelihood ratios, $\{\,j : t_j \ge \tau\,\}$: the exact local-FDR ranking coincides with the likelihood-ratio ranking, and strengthening the evidence for any one observation never ejects another from the discovery set.
\end{corollary}

\begin{proof}
	The leave-one-out identity \eqref{eq:responsibility} rearranges to $P(z_j = f \mid \mathpzc{X}) = t_j N_j / (A_j + t_j N_j)$, where $A_j = \sum_m e_m(\mathbf{t}_{-j}) B(m+\alpha, n-m+\beta)$ and $N_j = \sum_m e_m(\mathbf{t}_{-j}) B(m+\alpha+1, n-1-m+\beta)$ are positive and independent of $t_j$. This is strictly increasing in $t_j$, so $\mathrm{lfdr}_j$ is strictly decreasing in $t_j$; monotonicity in the remaining $t_i$ is the total positivity of \autoref{prop:monotone}.
\end{proof}

Finally, the exact posterior obeys the standard large-sample guarantees, approached through the log-concave family of \autoref{thm:logconcave} rather than from an approximation.

\begin{proposition}[Posterior contraction]\label{prop:bvm}
	Suppose the observations are i.i.d.\ from the two-groups density $p_{w_0} = w_0 f + (1-w_0) g$ with known components $f \ne g$ and $w_0 \in (0,1)$, and write $I(w_0) = \int \frac{\big(f(x) - g(x)\big)^2}{w_0 f(x) + (1-w_0) g(x)}\, dx \in (0,\infty)$ for the Fisher information. Then for any prior $\mathpzc{Beta}(\alpha,\beta)$ with $\alpha,\beta > 0$, the marginal posterior of $w$ satisfies a Bernstein--von Mises theorem: with $\hat w_n$ the maximum-likelihood estimator,
	\begin{equation}
		\sup_{A \subseteq \mathbb{R}}\left| P\big(\sqrt{n}\,(w - \hat w_n) \in A \mid \mathpzc{X}\big) - \mathcal{N}\!\big(0,\, I(w_0)^{-1}\big)(A) \right| \xrightarrow{P} 0,
	\end{equation}
	so that $\EX[w \mid \mathpzc{X}] \to w_0$ and $n\,\mathrm{Var}[w \mid \mathpzc{X}] \to I(w_0)^{-1}$.
\end{proposition}

\begin{proof}
	The single-observation log-likelihood $w \mapsto \log\big(w f(x) + (1-w) g(x)\big)$ is concave in $w$, being the logarithm of an affine function, so the one-parameter model is regular with Fisher information $I(w_0)$, and $w \mapsto p_w$ is identifiable since $f \ne g$. The claim is then the parametric Bernstein--von Mises theorem for a smooth, identifiable model under a prior that is positive and continuous at $w_0$ \citep[Thm.~10.1]{vandervaart1998}; the exact marginalisation evaluates this posterior without approximation.
\end{proof}

At the boundary $w_0 \in \{0,1\}$ -- the rare-signal regime of \autoref{sec:calibration} -- the interior argument does not apply and the limit is one-sided (a half-normal), which is precisely where symmetric Gaussian and Laplace intervals mis-cover while the exact posterior, being the true one, does not.

The contrast with the fast approximations of \autoref{sec:calibration} is instructive: an EM point estimate that has collapsed to a boundary, or a Laplace approximation built at a single mode, carries none of these guarantees -- it may be non-monotone in the data and assigns a degenerate interval. The exact posterior is log-concave, moment-complete, monotone, and asymptotically efficient by construction, and it is these properties, rather than agreement with a sampler, that distinguish it.

\section{Generalization to $K$-mixtures}\label{sec:kmix}

An algorithm with a similar complexity can be derived for a general case of $K > 2$ components, where $\bold{w} \sim \mathpzc{BL}(\bold{\alpha}, \gamma)$. To proceed further, we need the PDF of $\mathpzc{BL}$ (see \autoref{app:bl_pdf}), $q(\bold{w}|\bold{\alpha}, \gamma)$, and its mixed moments (see \autoref{app:bl_moments}):

\begin{equation}
q(\bold{w}|\bold{\alpha}, \gamma) = \frac{1}{B(\boldsymbol{{\hat{\alpha}}}) B(\alpha_{K}, \gamma)}  \left(\sum_{i=1}^{K-1} w_i \right) ^{\gamma - \sum_{i=1}^{K-1} \alpha_i} \prod_{i=1}^{K} w_i^{\alpha_i-1},
\end{equation}
\begin{equation}
		\EX\left[\prod_{i=1}^K w_i^{c_i}\right] = \frac{B(\boldsymbol{\hat{\alpha}} + \mathbf{\hat{c}})}{B(\boldsymbol{\hat{\alpha}})} \frac{B(\gamma + \sum_{i=1}^{K-1} c_i, \alpha_K + c_K)}{B(\gamma, \alpha_K)},
\end{equation}
where $B(\boldsymbol{\hat{\alpha}})$ is the multivariate Beta function, $B(\gamma, \alpha_K)$ is a bivariate Beta function, $\boldsymbol{\hat{\alpha}}$ is a subvector of $\boldsymbol{\alpha}$ without the last element $\alpha_K$, $\boldsymbol{c} \in \Re^K$ and $\boldsymbol{\hat{c}}$ is $\boldsymbol{c}$ without its last element $c_K$.

The integral of interest is:

\begin{align*}
	\int_{\Delta_{K-1}}\widetilde{L}(\mathpzc{X}|\theta, \mathbf{w}) q(\mathbf{w}|\boldsymbol\alpha, \gamma) d\bold{w} = \int_{\Delta_{K-1}} \prod_{i=1}^n \left( \sum_{j=1}^{K} w_j f_j(x_i) \right) q(\mathbf{w}|\boldsymbol\alpha, \gamma) d\bold{w},
\end{align*}
where $\int_{\Delta_{K-1}}$ is an integral over $K-1$ dimensional simplex.
Let's first expand $\widetilde{L}(\mathpzc{X}|\theta, \mathbf{w})$:

\begin{align*}
	\widetilde{L}(\mathpzc{X}|\theta, \mathbf{w}) = \prod_{i=1}^n \left( \sum_{j=1}^{K} w_j f_j(x_i) \right) = \prod_{i=1}^n \left( w_K f_K+ \sum_{j=1}^{K-1} w_j f_j(x_i) \right) = \prod_{i=1}^n \left( (1 - \sum_{j=1}^{K-1} w_j) f_K(x_i)+ \sum_{j=1}^{K-1} w_j f_j(x_i) \right) = \\ =
	\left(\underbrace{\prod_{i=1}^n f_K(x_i)}_{L_0(\mathpzc{X})} \right) \prod_{i=1}^n \left(  1+\sum_{j=1}^{K-1} w_j \underbrace{\left( \frac{f_j(x_i)}{f_K(x_i)} - 1\right)}_{t_j(x_i)} \right) = L_0(\mathpzc{X}) \prod_{i=1}^n \left(1 + \sum_{j=1}^{K-1} w_j t_j(x_i) \right).
\end{align*}
At this point, we have ruled out contributions related to the last $w_K$ component: since $w_K$ has been eliminated, every monomial in the expansion has $c_K = 0$. Its total degree $m = \sum_{i=1}^{K-1} c_i$ is not fixed but ranges over $0, \dots, n$ (each of the $n$ factors contributes either the constant $1$ or a $w_j t_j$ term), so the corresponding moment is proportional to $\frac{B\left(\gamma + m, \alpha_K\right)}{B(\gamma, \alpha_K)}$. This allows us to work with regular $\mathpzc{Dir}$ moments rather than convoluted $\mathpzc{BL}$ moments. Furthermore, we can notice that the moment function now takes a peculiar form:
\begin{align*}
\EX\left[\prod_{i=1}^{K-1} w_i^{c_i}\right] = \frac{B(\gamma + m, \alpha_K )}{B(\gamma, \alpha_K)} \frac{B(\boldsymbol{\hat{\alpha}} + \mathbf{\hat{c}})}{B(\boldsymbol{\hat{\alpha}})} = \frac{B(\gamma + m, \alpha_K )}{B(\gamma, \alpha_K) B(\boldsymbol{\hat{\alpha}}) \Gamma(m + \sum_{i=1}^{K-1} \alpha_i)} \prod_{i=1}^{K-1}\Gamma(\alpha_i + c_i) = \\ = c_m(\boldsymbol \alpha, \gamma) \prod_{i=1}^{K-1}\Gamma(\alpha_i + c_i).
\end{align*}
This structure hints to us that the contribution of expectation of each monomial can be decomposed multiplicatively as a product of independent contributions $\Gamma(\alpha_i + c_i)$.

Then, we reformulate the product as 

$$
\prod_{i=1}^n \left( 1+ \sum_{j=1}^{K-1} w_j t_j(x_i) \right) = \sum_{S \subseteq \left\{1, \dots, n\right\}} \prod_{i \in S} \left( \sum_{j=1}^{K-1} w_j t_j(x_i) \right) = \sum_{m=0}^n \sum_{\substack{S \subseteq \{1, \dots, n\} \\ |S|=m}} \prod_{i \in S} \left( \sum_{j=1}^{K-1} w_j t_j(x_i) \right),
$$
where $S$ is a subset of set $\{1, \dots, n\}$, i.e. each subset $S$ corresponds to a pattern of assignments of data samples to the first $K-1$ components:

\begin{itemize}
	\item  $ |S| = m$  $m$ data points are assigned to components $ \{1, \dots, K-1\} $;
	\item For complement $|\bar{S}| = n - m $: $ n - m $ data points are assigned to the $K$-th component. This case is absorbed by the $1$ term.
\end{itemize}

\begin{table}[th]
\begin{tcolorbox}[
 colback=white, 
colframe=gray, 
boxrule=0.5pt, 
arc=0mm, 
boxsep=5pt, 
left=5pt, 
right=5pt, 
top=5pt, 
bottom=5pt, 
title=$S$ expansion example, 
coltitle=black, 
fonttitle=\normalfont, 
colbacktitle=gray!20, 
enhanced, 
attach boxed title to top center={yshift=-2mm} 
	]
Although the notation is clear, a simple example should clear up any possible confusion.
For simplicity, let $K = 2$ (two components) and $n = 2$ data points. The expansion becomes:
$$
\prod_{i=1}^2 \left(1 + w_1 t_1(x_i)\right) = \sum_{S \subseteq \{1, 2\}} \prod_{i \in S} w_1 t_1(x_i).
$$
Explicitly, the subsets $S$ are:
\begin{itemize}
	\item $S = \emptyset$: No data points assigned to component 1. Contribution: $1$.
	\item $S = \{1\}$: Assign $x_1$ to component 1. Contribution: $w_1 t_1(x_1)$.
	\item $S = \{2\}$: Assign $x_2$ to component 1. Contribution: $w_1 t_1(x_2)$.
	\item $S = \{1, 2\}$: Assign both $x_1$ and $x_2$ to component 1. Contribution: $w_1^2 t_1(x_1) t_1(x_2)$.
\end{itemize}

Thus, the expanded formula is:
$$
\prod_{i=1}^2 \left(1 + w_1 t_1(x_i)\right) = 1 + w_1 t_1(x_1) + w_1 t_1(x_2) + w_1^2 t_1(x_1) t_1(x_2).
$$
\end{tcolorbox}
\end{table}

Expanding each inner factor $\left(\sum_{j=1}^{K-1} w_j t_j(x_i)\right)$ amounts to choosing, for every $i \in S$, a single component $\phi(i) \in \{1, \dots, K-1\}$ to which observation $i$ is assigned; the contribution is $\prod_{i \in S} w_{\phi(i)} t_{\phi(i)}(x_i)$. Grouping by the count vector $\boldsymbol{c} = (c_1, \dots, c_{K-1})$, $c_j = |\{i \in S : \phi(i) = j\}|$, gives
\begin{equation}\label{eq:kmix_multivar}
	\prod_{i=1}^n \left( 1 + \sum_{j=1}^{K-1} w_j t_j(x_i) \right) = \sum_{\boldsymbol{c} \ge \boldsymbol{0}} E_{\boldsymbol{c}}(\mathbf{t})\, \prod_{j=1}^{K-1} w_j^{c_j},
\end{equation}
where the coefficient is the \textit{multivariate elementary symmetric polynomial}
\begin{equation}\label{eq:multiv_esp}
	E_{\boldsymbol{c}}(\mathbf{t}) = \sum_{\substack{S_1, \dots, S_{K-1}~\text{pairwise disjoint} \\ S_j \subseteq \{1,\dots,n\},~|S_j| = c_j}} \prod_{j=1}^{K-1} \prod_{i \in S_j} t_j(x_i).
\end{equation}

\begin{remark}[Non-factorisation across components]\label{rem:correction}
	The coefficient $E_{\boldsymbol{c}}(\mathbf{t})$ does \emph{not} separate into a product of univariate elementary symmetric polynomials for $K \ge 3$. Writing $e_j(m) = \sum_{|S|=m} \prod_{i\in S} t_j(x_i)$, the product $\prod_{j=1}^{K-1} e_j(c_j)$ ranges over subsets $S_1,\dots,S_{K-1}$ chosen \emph{independently}, whereas \eqref{eq:multiv_esp} constrains the $S_j$ to be \emph{pairwise disjoint} -- each observation contributes exactly one factor $\left(\sum_j w_j t_j(x_i)\right)$ and is thus assigned to at most one component. The distinction is already visible at $K=3$, $n=2$: writing $t_j(x_i) = t_{ji}$, the coefficient of $w_1 w_2$ is $t_{11}t_{22} + t_{12}t_{21}$, whereas $e_1(1)\,e_2(1) = (t_{11}+t_{12})(t_{21}+t_{22})$ carries in addition the terms $t_{11}t_{21}$ and $t_{12}t_{22}$, in which a single observation is assigned to two components at once. Because $E_{\boldsymbol{c}}$ does not factor over components, no per-component convolution can assemble the marginal for $K \ge 3$; the joint dynamic program of \eqref{eq:joint_dp} does. For $K = 2$ there is a single non-reference component, disjointness is automatic, $E_{(c_1)} = e_1(c_1)$, and the two-component algorithm of \autoref{sec:twocomp} is recovered.
\end{remark}

Substituting \eqref{eq:kmix_multivar} and integrating monomial by monomial with the moment identity above yields
\begin{align*}
	\int_{\Delta_{K-1}}\widetilde{L}(\mathpzc{X}|\theta, \mathbf{w}) q(\mathbf{w}|\boldsymbol\alpha, \gamma) d\bold{w}
	&= L_0(\mathpzc{X}) \sum_{m=0}^n \sum_{\sum_j c_j = m} E_{\boldsymbol{c}}(\mathbf{t})\, \EX\!\left[\prod_{j=1}^{K-1} w_j^{c_j}\right] \\
	&= L_0(\mathpzc{X}) \sum_{m=0}^n c_m(\boldsymbol{\alpha}, \gamma) \underbrace{\sum_{\sum_j c_j = m} E_{\boldsymbol{c}}(\mathbf{t}) \prod_{j=1}^{K-1} \Gamma(\alpha_j + c_j)}_{G(m)}.
\end{align*}

\paragraph{A numerically robust equivalent.} Factoring out the reference component makes the $t_j(x_i) = f_j(x_i)/f_K(x_i) - 1$ (and hence the $E_{\boldsymbol{c}}$) signed, which invites catastrophic cancellation when $f_K$ is not the dominant density. It is cleaner to keep all $K$ densities and expand the likelihood directly,
\begin{equation}\label{eq:kmix_positive}
	\prod_{i=1}^n \left( \sum_{j=1}^{K} w_j f_j(x_i) \right) = \sum_{\substack{\boldsymbol{c} \in \mathbb{Z}_{\ge0}^K \\ |\boldsymbol{c}| = n}} M(\boldsymbol{c})\, \prod_{j=1}^{K} w_j^{c_j}, \qquad M(\boldsymbol{c}) = \sum_{\substack{a : \{1,\dots,n\} \to \{1,\dots,K\} \\ \mathrm{counts}(a) = \boldsymbol{c}}} \prod_{i=1}^n f_{a_i}(x_i),
\end{equation}
so that all $M(\boldsymbol{c}) \ge 0$ and, by \eqref{eq:bl_mixed_moments},
\begin{equation}\label{eq:kmix_final}
	\int_{\Delta_{K-1}}\widetilde{L}(\mathpzc{X}|\theta, \mathbf{w}) q(\mathbf{w}|\boldsymbol\alpha, \gamma)\, d\bold{w} = \sum_{\substack{|\boldsymbol{c}| = n}} M(\boldsymbol{c})\, \EX\!\left[\prod_{j=1}^{K} w_j^{c_j}\right].
\end{equation}
Every summand is positive, so the whole computation may be carried out in the log-domain with $\log$-$\mathrm{sum}$-$\exp$ and never cancels.

\subsection{Computing the coefficient array}\label{sect:zeta}

The array $\{M(\boldsymbol{c})\}$ is the coefficient tensor of the product of $n$ affine forms $\sum_{j=1}^K w_j f_j(x_i)$ in \eqref{eq:kmix_positive}. It is built by processing observations one at a time, extending the count vector by one unit:
\begin{equation}\label{eq:joint_dp}
	M^{(i)}(\boldsymbol{c}) = \sum_{j=1}^{K} f_j(x_i)\, M^{(i-1)}(\boldsymbol{c} - \boldsymbol{e}_j), \qquad M^{(0)}(\boldsymbol{0}) = 1,
\end{equation}
where $\boldsymbol{e}_j$ is the $j$-th unit vector and terms with a negative index are $0$. After $n$ steps, $\{M(\boldsymbol{c})\}_{|\boldsymbol{c}|=n}$ is complete, and \eqref{eq:kmix_final} follows in one weighted sum. Equation \eqref{eq:joint_dp} is the multivariate analogue of the two-component recurrence $e(m; n) = e(m; n-1) + t_n\, e(m-1; n-1)$: at each observation one either assigns it to component $K$ (index unchanged) or to some $j < K$ (index incremented). It is spelled out in \autoref{sect:H}; in the log-domain each update is a $\log$-$\mathrm{sum}$-$\exp$ over the $K$ predecessors, which is unconditionally stable.

\begin{remark}[Complexity]\label{rem:kcomplexity}
	The number of count vectors with $|\boldsymbol{c}| \le n$ over $K$ components is $\binom{n+K-1}{K-1} = O(n^{K-1})$; each of the $n$ observations touches every state with $O(K)$ work, giving $O(K\,n^{K})$ time and $O(n^{K-1})$ storage. For $K = 2$ this collapses to the $O(n^2)$ two-component dynamic program (and the $O(n\log^2 n)$ greedy-FFT variant), whereas for $K = 3$ it is $O(n^3)$ time and $O(n^2)$ storage. The dependence on the count-vector support is intrinsic -- $E_{\boldsymbol c}$ does not factor across components (Remark~\ref{rem:correction}) -- so no per-component convolution can reduce it; the scheme is nonetheless very practical for the small $K$ ($2$ or $3$) that arise in the applications of the introduction, especially within a hierarchy of moderate-sized strata.
\end{remark}

The two elementary-symmetric-polynomial routines below (dynamic programming and greedy FFT) are the $K=2$ specialisation of \eqref{eq:joint_dp}: they remain the workhorse for the two-groups model and are the innermost building block for larger $K$. The underlying univariate recurrence is
$e(m; n) = e(m; n-1) + t_n \cdot e(m-1; n-1)$, obtained by partitioning subsets $S$ into those excluding $t_n$ (contributing $e(m; n-1)$) and those including it (contributing $t_n\, e(m-1; n-1)$):

\begin{algorithm}
	\caption{Compute elementary symmetric polynomials via DP}
	\begin{algorithmic}[1]
		\Function{CalculateESP}{$m$, $n$, $t_1, \ldots, t_n$}
		\State Initialize 2D array $dp$ of size $(m+1) \times (n+1)$ with $0$s
		\For{$k = 0$ to $n$}
		\State $dp[0][k] \gets 1$ \Comment{Base case: $e(0; k) = 1$}
		\EndFor
		\For{$j = 1$ to $n$}
		\For{$i = 1$ to $\min(m, j)$}
		\State $dp[i][j] \gets dp[i][j-1] + t_j \cdot dp[i-1][j-1]$ \Comment{Recurrence}
		\EndFor
		\EndFor
		\State \Return $dp[m][n]$
		\EndFunction
	\end{algorithmic}
\end{algorithm}
\begin{algorithm}
	\caption{Compute elementary symmetric polynomials via greedy FFT}
	\begin{algorithmic}[1]
		\Procedure{CalculateESP}{$\boldsymbol{t}$}
		\State \textbf{Input:} Real vector $\boldsymbol{t}$ of size $n$
		\State \textbf{Output:} Coefficients of $P(z) = \prod_{i=1}^{n} (1 + t_i z)$, i.e. the elementary symmetric polynomials
		\State \textbf{Initialize:} Create an empty min-heap $Q$ (priority based on polynomial degree)
		\For{$i \gets 1$ to $n$}
		\State $P_i(z) \gets 1 + t_i z$
		\State \textbf{Insert} $P_i(z)$ into $Q$
		\EndFor
		\While{$|Q| > 1$}
		\State $A(z) \gets$ \textbf{Extract-Min}$(Q)$
		\State $B(z) \gets$ \textbf{Extract-Min}$(Q)$
		\State $C(z) \gets$ \Call{PolyMult}{$A(z), B(z)$} \Comment{Multiply polynomials, see \autoref{app:sect_fft_asstmptotic}}
		\State \textbf{Insert} $C(z)$ into $Q$
		\EndWhile
		\State \textbf{return} \textbf{Extract-Min}$(Q)$ \Comment{Final polynomial $P(z)$}
		\EndProcedure
	\end{algorithmic}
\end{algorithm}

\subsection{The general $K$-component algorithm}\label{sect:H}

For $K \ge 3$ the marginal is assembled by the joint dynamic program \eqref{eq:joint_dp}, since the coefficient array does not factor across components (Remark~\ref{rem:correction}). The pseudocode below carries the counts of the first $K-1$ components explicitly; the count of the reference component $K$ is implied by the number of observations processed. Working in the log-domain keeps every partial sum positive and the computation unconditionally stable.

\begin{algorithm}[H]
	\caption{Exact $K$-component marginal via the joint count DP (log-domain)}
	\begin{algorithmic}[1]
		\Function{KMixLogMarginal}{$F \in \mathbb{R}^{n\times K}$, $\boldsymbol{\alpha}$, $\gamma$}
		\State $A[\boldsymbol{c}] \gets -\infty$ for all $\boldsymbol{c} \in \{0,\dots,n\}^{K-1}$; \quad $A[\boldsymbol{0}] \gets 0$ \Comment{$A[\boldsymbol{c}] = \log M(\boldsymbol{c})$}
		\For{$i = 1$ to $n$}
		\State $A'[\boldsymbol{c}] \gets A[\boldsymbol{c}] + \log F_{i,K}$ \Comment{assign obs.\ $i$ to reference comp.\ $K$}
		\For{$j = 1$ to $K-1$}
		\State $A'[\boldsymbol{c} + \boldsymbol{e}_j] \gets \operatorname{logaddexp}\!\big(A'[\boldsymbol{c} + \boldsymbol{e}_j],\; A[\boldsymbol{c}] + \log F_{i,j}\big)$ \Comment{assign to comp.\ $j$}
		\EndFor
		\State $A \gets A'$
		\EndFor
		\State \Return $\operatorname{logsumexp}_{|\boldsymbol{c}| \le n}\big( A[\boldsymbol{c}] + \log \EX[\textstyle\prod_j w_j^{c_j}] \big)$ \Comment{$c_K = n - |\boldsymbol{c}|$, moment via \eqref{eq:bl_mixed_moments}}
		\EndFunction
	\end{algorithmic}
\end{algorithm}

When the $t_j(x_i)$ are known to be well behaved (e.g.\ a dominant reference component, or the two-groups case), the signed form $Z = L_0 \sum_m c_m(\boldsymbol{\alpha},\gamma) G(m)$ with $G(m)$ computed from \eqref{eq:multiv_esp} is an equivalent, sometimes faster, route; the log-domain program above is the default because it never cancels.

\section{Numerical simulation}\label{sec:numerical}

We validate the exact algorithms and quantify their behaviour on both synthetic and real data. A reference Python implementation reproducing every figure and table of this section is provided (\autoref{sec:reproducibility}); throughout, ``exact'' refers to the log-domain algorithms of Sections \ref{sec:twocomp}--\ref{sec:kmix}.

\subsection{Correctness}

Table~\ref{tab:correctness} certifies the exact marginal against independent ground truth. For small $n$ we compare with the naive $2^n$ enumeration and with adaptive and $60$-digit quadrature of the defining integral; the paper's own $S_{\mathbf{t}}$ recurrences (\autoref{eq:marg_rec1}, \autoref{eq:marg_ab}) are cross-checked against their brute-force definition. For the $K$-component program we compare with a $K^n$ brute-force enumeration and with Monte-Carlo integration over the simplex. Agreement is at the level of floating-point round-off wherever an independent exact reference exists, and within Monte-Carlo error otherwise.

\begin{table}[H]
	\centering
	\small
	\begin{tabular}{@{}llll@{}}
		\toprule
		Quantity & Reference & Regime & Max.\ relative error \\
		\midrule
		2-comp.\ marginal $Z$ & $2^n$ enumeration & $n \le 20$ & $1.7\times10^{-12}$ \\
		2-comp.\ marginal $Z$ & $60$-digit ESP sum (mpmath) & $n \le 3000$ & $7.3\times10^{-12}$ (log $Z$) \\
		$S_{\mathbf{t}}$ recurrences & brute-force definition & $n \le 9$ & $<10^{-15}$ \\
		posterior mean/var of $w$ & fine-grid quadrature & $n = 40$ & $<10^{-4}$ \\
		local FDR (per obs.) & leave-one-out quadrature & $n = 12$ & $<10^{-5}$ \\
		$K$-comp.\ marginal $Z$ & $K^n$ enumeration & $K{=}3, n\le 8$ & $<10^{-12}$ (log $Z$) \\
		$K$-comp.\ marginal $Z$ & Monte-Carlo (simplex) & $K{=}3, n = 18$ & $<2\%$ (MC error) \\
		\bottomrule
	\end{tabular}
	\caption{Validation of the exact algorithms against independent ground truth, as reproduced by the accompanying validation suite (\autoref{sec:reproducibility}). All checks pass at the accuracy the reference permits.}
	\label{tab:correctness}
\end{table}

\subsection{Scaling and numerical stability}

\autoref{fig:scaling} contrasts the exact marginal with the intractable $2^n$ enumeration and reports the scaling of the three backends. The naive integral is already infeasible by $n=20$ (over $30$~s), whereas the exact log-domain evaluation takes well under a millisecond there and scales as $O(n^2)$ thereafter (about $8$~ms at $n=10^3$ and $0.2$~s at $n=6033$), remaining accurate for every $n$ tested. The plain-domain DP and the greedy FFT are faster per operation but are limited by the plain-domain elementary symmetric polynomials, whose magnitudes grow combinatorially ($e_m \sim \binom{n}{m}$, overflowing double precision near $n \approx 10^{3}$) and whose dynamic range explodes for well-separated components; the log-domain DP is therefore the only backend we recommend for evaluating the marginal, and at very large $n$ its $O(n^2)$ cost -- not any accuracy deficit -- is the binding constraint. Its accuracy is validated to $\sim\!10^{-12}$ relative error against a $60$-digit reference for every $n$ and every separation tested (\autoref{tab:correctness}); the plain-domain FFT backend, by contrast, returns finite but grossly wrong values well before it overflows, and we document that failure -- the reason the FFT is retained only as a complexity result -- in \autoref{app:sect_fft_asstmptotic}.

\begin{figure}[H]
	\centering
	\includegraphics[width=\textwidth]{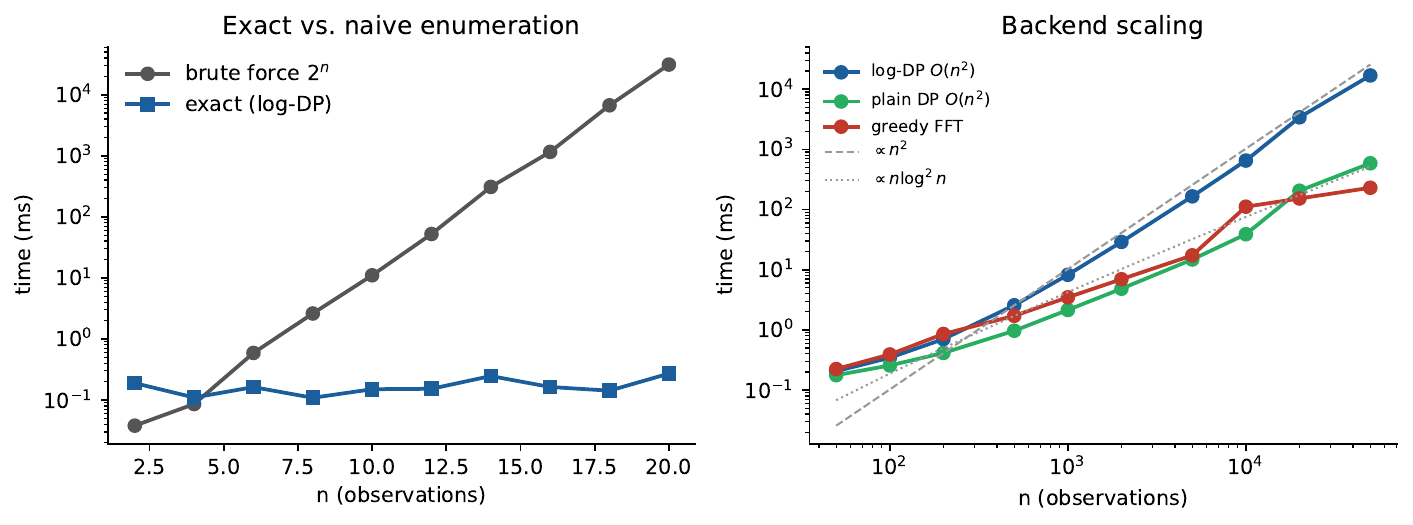}
	\caption{Left: the exact marginal (log-domain DP) versus naive $2^n$ enumeration; the latter is already $>30$~s at $n=20$. Right: backend scaling; the log-domain DP is $O(n^2)$, the greedy FFT approaches $O(n\log^2 n)$.}
	\label{fig:scaling}
\end{figure}

\subsection{Comparison with EM and MCMC}

The exact posterior of $w$ is a finite mixture of Beta distributions (\autoref{sec:posterior}), returned deterministically in one pass together with the marginal likelihood. \autoref{fig:mcmc} compares it, on a two-groups problem, with a data-augmentation Gibbs sampler (whose stationary law is exactly $p(w\mid\mathpzc{X})$), with PyMC's NUTS, and with MAP--EM. The exact posterior mean, standard deviation, and $95\%$ credible interval agree with both samplers to three or four significant figures (e.g.\ at $n=400$: exact mean $0.2839$, sd $0.0270$, versus Gibbs $0.2838$/$0.0272$ and NUTS $0.2837$/$0.0272$), and the per-observation posterior probabilities match the Gibbs estimates to $1.5\times10^{-4}$. This agreement is a self-consistency check, not an accuracy advantage: a converged sampler targets the same posterior. The operational gains are that the exact computation is deterministic, needs no burn-in, effective-sample-size or Gelman--Rubin diagnostics, and -- unlike EM, which returns only a point estimate -- delivers full uncertainty quantification. Its speed advantage is largest at small and moderate $n$ ($\sim\!170\times$ versus Gibbs at $n=100$, $\sim\!40\times$ versus NUTS at $n=400$) and narrows as the $O(n^2)$ cost grows, falling to $\sim\!3\times$ versus Gibbs by $n=6000$; for a single very large stratum a sampler is eventually the faster choice.

\begin{figure}[H]
	\centering
	\includegraphics[width=\textwidth]{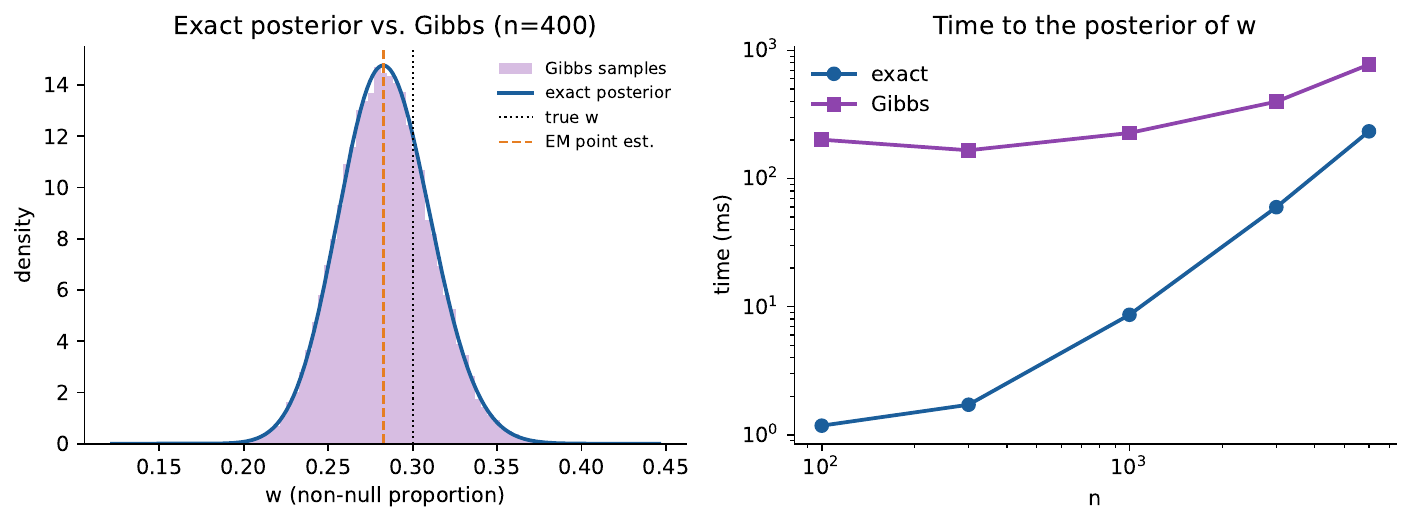}
	\caption{Left: exact Beta-mixture posterior of $w$ (curve) against $50{,}000$ post-burn-in Gibbs draws (histogram); the EM point estimate carries no uncertainty. Right: time to obtain the posterior of $w$, exact versus Gibbs.}
	\label{fig:mcmc}
\end{figure}

\subsection{Calibrated uncertainty where fast approximations fail}\label{sec:calibration}

Because the exact method and a sampler both target $p(w\mid\mathpzc{X})$, the exact method cannot be ``more accurate'' than a converged MCMC run. Its advantage is over the \emph{cheap} approximations that practitioners actually reach for -- a Gaussian/Laplace posterior, an EM point estimate with a Wald interval, or a moment-matched Beta -- whose intervals are not trustworthy in the small-sample, rare-signal regime. To measure this we ran a frequentist study: for each sample size $n$ and true weight $w$ we draw $R=300$ two-groups datasets, form the nominal $95\%$ interval four ways, and record both its coverage and its total-variation distance to the exact posterior (\autoref{fig:calibration}). The exact equal-tailed credible interval is calibrated across the whole grid (mean coverage $0.952$, range $[0.927,0.977]$) and, being the truth, has zero shape error. The approximations fail in different ways. The logit-Laplace interval \emph{under-covers} outright -- $0.78$ at $n=10$, $w=0.05$, about $12$ standard errors (of the exact-method coverage estimate) below. The EM+Wald and moment-matched-Beta intervals instead keep their nominal \emph{containment count}, but only by being the wrong shape: EM+Wald clips a symmetric Wald interval at $0$, inflating its standard deviation roughly two-fold, while the moment-matched Beta is outright infeasible -- its local variance exceeds $w(1-w)$, admitting no unimodal Beta -- in $59\%$ of the hardest cells, and where it is defined it sits \emph{further} from the exact posterior than the Laplace approximation ($\mathrm{TV}=0.33$ versus $0.13$ at $n=10$, $w=0.05$). So coverage alone is a weak diagnostic: only the exact posterior is simultaneously calibrated and shape-correct. Where the posterior is near-Gaussian ($w=0.5$, or large $n$) all four methods agree, as asymptotics require; the exact method's edge is concentrated in the skewed, boundary-bounded regime -- small strata with rare signals -- that dominates large-scale multiple testing and multi-study pooling.

\begin{figure}[H]
	\centering
	\includegraphics[width=\textwidth]{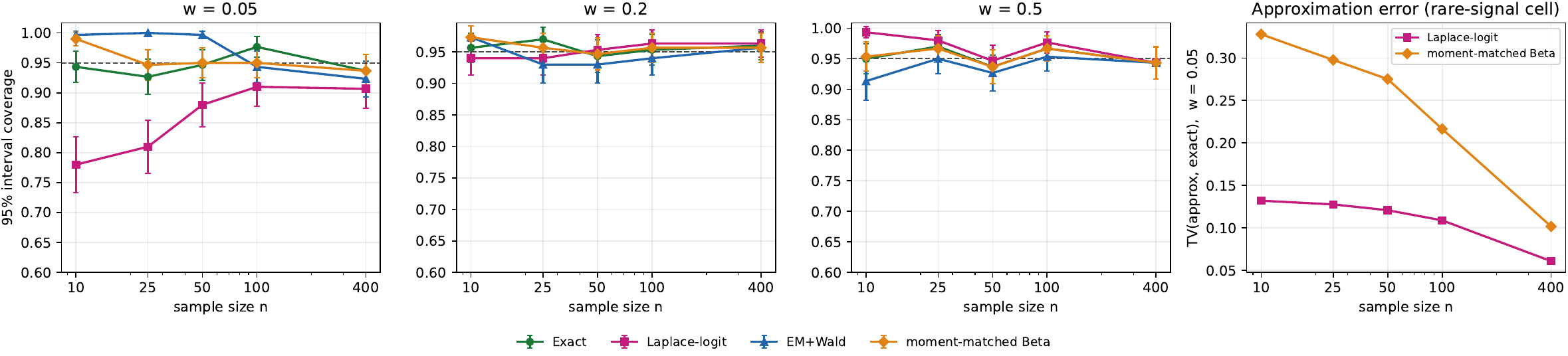}
	\caption{Frequentist calibration of the $95\%$ interval for the mixing weight ($R=300$ datasets per cell, two-groups Gaussian, flat prior). First three panels: coverage versus $n$ at $w=0.05,0.2,0.5$. The exact interval holds nominal coverage throughout; the logit-Laplace interval drops to $0.78$ at $n=10$, $w=0.05$; EM+Wald and the moment-matched Beta reach nominal coverage only through over-wide or (for the Beta, $59\%$ of the time) infeasible intervals. Right: total-variation distance of the two Beta-support approximations to the exact posterior in the rare-signal cell -- the moment-matched Beta is \emph{further} from the truth than the Laplace approximation despite comparable coverage, confirming that coverage alone does not certify an interval.}
	\label{fig:calibration}
\end{figure}

Because the interval is Bayesian, its frequentist coverage in principle depends on the prior, so we check that the flat-prior result is not a convenient one. \autoref{tab:prior_sensitivity} recomputes the exact interval's coverage across the grid under a Jeffreys $\mathpzc{Beta}(\tfrac12,\tfrac12)$ and a deliberately informative rare-signal $\mathpzc{Beta}(1,4)$ (prior mean $0.2$). The mean coverage stays at $0.95$ under all three priors; the flat and Jeffreys priors hold near-nominal coverage in every cell, and the informative prior even sharpens the rare-signal corner ($n{=}10$, $w{=}0.05$: $0.94\!\to\!0.98$), deforming coverage only at $n{=}10$, $w{=}0.5$ ($0.89$), the one cell where its mean is far from the truth. The calibration of \autoref{fig:calibration} is thus a property of the exact posterior, not an artefact of the prior.

\begin{table}[H]
	\centering
	\begin{tabular}{lccc}
		\toprule
		prior on $w$ & mean coverage & worst cell & rare-signal cell \\
		 & (whole grid) & (min over grid) & ($n{=}10$, $w{=}0.05$) \\
		\midrule
		flat $\mathpzc{Beta}(1,1)$ & $0.953$ & $0.927$ & $0.943$ \\
		Jeffreys $\mathpzc{Beta}(\tfrac{1}{2},\tfrac{1}{2})$ & $0.955$ & $0.930$ & $0.980$ \\
		informative $\mathpzc{Beta}(1,4)$ & $0.957$ & $0.890$ & $0.983$ \\
		\bottomrule
	\end{tabular}
	\caption{Prior sensitivity of the exact $95\%$ credible interval's frequentist coverage, over a $4\times3$ grid of sample sizes $n\in\{10,25,50,100\}$ and weights $w\in\{0.05,0.2,0.5\}$ ($R=300$ datasets per cell). Coverage is near nominal for the flat and Jeffreys priors throughout; only the deliberately informative $\mathpzc{Beta}(1,4)$ under-covers, and only at its worst cell $n{=}10$, $w{=}0.5$, where its prior mean $0.2$ is farthest from the truth.}
	\label{tab:prior_sensitivity}
\end{table}

Two caveats bound this advantage. First, at an actual boundary ($w=0$ or $1$) with fixed components the maximum-likelihood/EM estimate correctly returns the boundary, and it is the \emph{Bayesian} interval that then under-covers, because the prior pulls it into the interior; the boundary is a caution, not a place where the exact method wins. Second, the gain concerns the \emph{posterior of the proportion} and does not automatically transfer to every downstream summary. The per-observation local FDR depends on $w$ only through the mild ratio $wf/(wf+(1-w)g)$, so integrating over the posterior of $w$ and substituting its point estimate give almost identical answers: in a matched simulation ($f=\mathcal N(3,1)$, $g=\mathcal N(0,1)$, $w=0.15$, $R=300$) the exact and Efron-style plug-in local FDR realise the same false-discovery rate (to three decimals) and nearly the same power (agreeing to two decimals, and differing only in the third) for every $n\ge 50$. The exact machinery is worth the effort for the global proportion and its uncertainty, for model evidence, and for empirical-Bayes fitting (below), but not for per-item thresholding, where a point estimate already suffices.

\subsection{Application I: signal prevalence in a small-study meta-analysis}\label{sec:app_meta}

The regime in which the calibration advantage of \autoref{sec:calibration} pays off is \emph{small strata}, and a multilevel meta-analysis is a clean real-world instance of the model of \autoref{fig:model}. We take the education compendium of \citet{Konstantopoulos2011}, distributed as \texttt{metadat::dat.konstantopoulos2011}: $56$ standardised school-level effects of an intervention on academic achievement, nested in $11$ districts. Within a district the schools are a null/non-null mixture; the target is $w_m$, the proportion of that district's schools with a genuine effect. Each district has only $3$--$11$ schools -- a regime in which a point estimate is nearly useless.

\autoref{fig:app_meta} shows the result. With a flat prior (so that EM is maximum likelihood), \emph{EM collapses onto a boundary} -- reporting $w_m$ at, or within a per-cent of, $0$ or $1$, i.e.\ ``no school'' or ``essentially every school'' has an effect, with no usable interval -- for $8$ of the $11$ districts, precisely because a handful of all-null or all-significant effects maximise the likelihood at the edge. The exact posterior instead returns a proper, appropriately wide $95\%$ credible interval for every district (all $11$ in $9$~ms total), correctly conveying that with $n=3$--$11$ the prevalence is barely identified rather than certainly $0$ or $1$. The interval is not itself an approximation: a data-augmentation Gibbs sampler reproduces the exact posterior to three figures (largest district: exact mean $0.312$ vs Gibbs $0.310$) but takes $270$~ms against the exact $0.9$~ms -- a $300\times$ gap \emph{per district} that becomes decisive across a compendium of thousands. Fitting a shared prior to pool across districts is only weakly identified at $11$ strata; the controlled partial-pooling demonstration is the synthetic study of \autoref{fig:hierarchical}.

\begin{figure}[H]
	\centering
	\includegraphics[width=\textwidth]{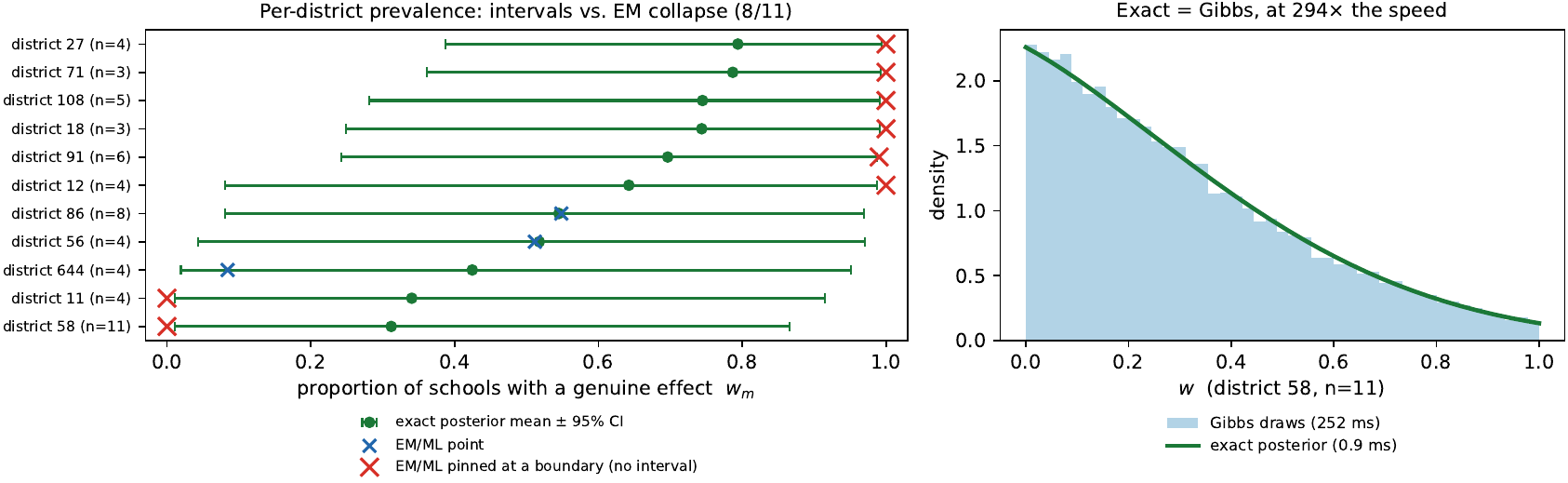}
	\caption{Real hierarchical meta-analysis (Konstantopoulos 2011, $11$ school districts). Left: per-district proportion $w_m$ of schools with a genuine effect -- exact posterior mean with $95\%$ credible interval (green) versus the EM/ML point, which is pinned at a boundary and yields no interval (red) for $8$ of $11$ districts. Right: for the largest district the exact posterior (curve) matches $32{,}000$ post-burn-in Gibbs draws (histogram) but is computed in $0.9$~ms against the sampler's $270$~ms ($300\times$).}
	\label{fig:app_meta}
\end{figure}

\subsection{Application II: pathway-level dysregulation prevalence}\label{sec:app_pathway}

The most direct biological use is to read off, for a curated gene set, what fraction of its genes are dysregulated between two conditions -- an effect-size-like quantity that the usual over-representation and enrichment tests, which report only an enrichment $p$-value against a hard significance cut-off, never estimate. We map the $7129$ probes of the Golub leukemia Hu6800 array \citep{Golub1999} to gene symbols through the GEO platform annotation (GPL80), collapse to $5171$ genes, and group them by the $50$ MSigDB Hallmark pathways \citep{Liberzon2015}. Because ALL versus AML is a very strong global contrast, the bulk of the $z$-values is over-dispersed; we therefore adopt Efron's empirical null \citep{Efron2004} (robust centre $-0.25$, scale $1.73$) so that ``dysregulated'' means beyond the typical gene, and fit the non-null scale once by evidence ($\hat\sigma_1=1.62$). The $50$ pathways are exchangeable strata, so rather than assume a prior on their prevalences we \emph{fit} a shared $\mathpzc{Beta}(a,b)$ by the same evidence maximisation across all pathways at once (empirical Bayes), obtaining $\mathpzc{Beta}(1.80, 2.03)$ (mean $0.47$); each pathway's prevalence is then its exact posterior under this fitted prior.

Each pathway's partially pooled posterior then gives its dysregulation prevalence with a closed-form $95\%$ interval (\autoref{fig:app_pathway}). The ranking is biologically coherent: the most dysregulated Hallmark pathways are \textsc{e2f\_targets} ($0.86$, CI $[0.66,0.98]$), \textsc{myc\_targets} ($0.82$), \textsc{tnfa\_signaling\_via\_nfkb} ($0.79$), \textsc{mtorc1\_signaling} and \textsc{il6\_jak\_stat3\_signaling} ($0.71$) -- the proliferation and inflammatory-signalling programs that distinguish the two leukemias -- while tissue-irrelevant sets (\textsc{spermatogenesis}, \textsc{uv\_response\_dn}) sit near $0.17$. Two advantages over the alternatives are visible. First, EM saturates: for the strongest pathways it reports a bare $w=1$ (``every gene dysregulated'') -- for $5$ of the $50$ sets -- whereas the pooled posterior reads, e.g., $0.86\,[0.66,0.98]$. Second, the prior is fitted, not assumed: the pathways shrink toward the evidence-fitted $\mathpzc{Beta}(1.80, 2.03)$ (mean $0.47$), and the smallest sets shrink most, so a prevalence read off a handful of genes -- \textsc{angiogenesis} at $n=26$, say -- is regularised toward the corpus rather than taken at face value, while still carrying the widest interval ($[0.16, 0.88]$). The whole $50$-pathway fit is deterministic and runs in milliseconds; that its intervals are trustworthy is certified by the known-truth benchmark of the next section.

\begin{figure}[H]
	\centering
	\includegraphics[width=\textwidth]{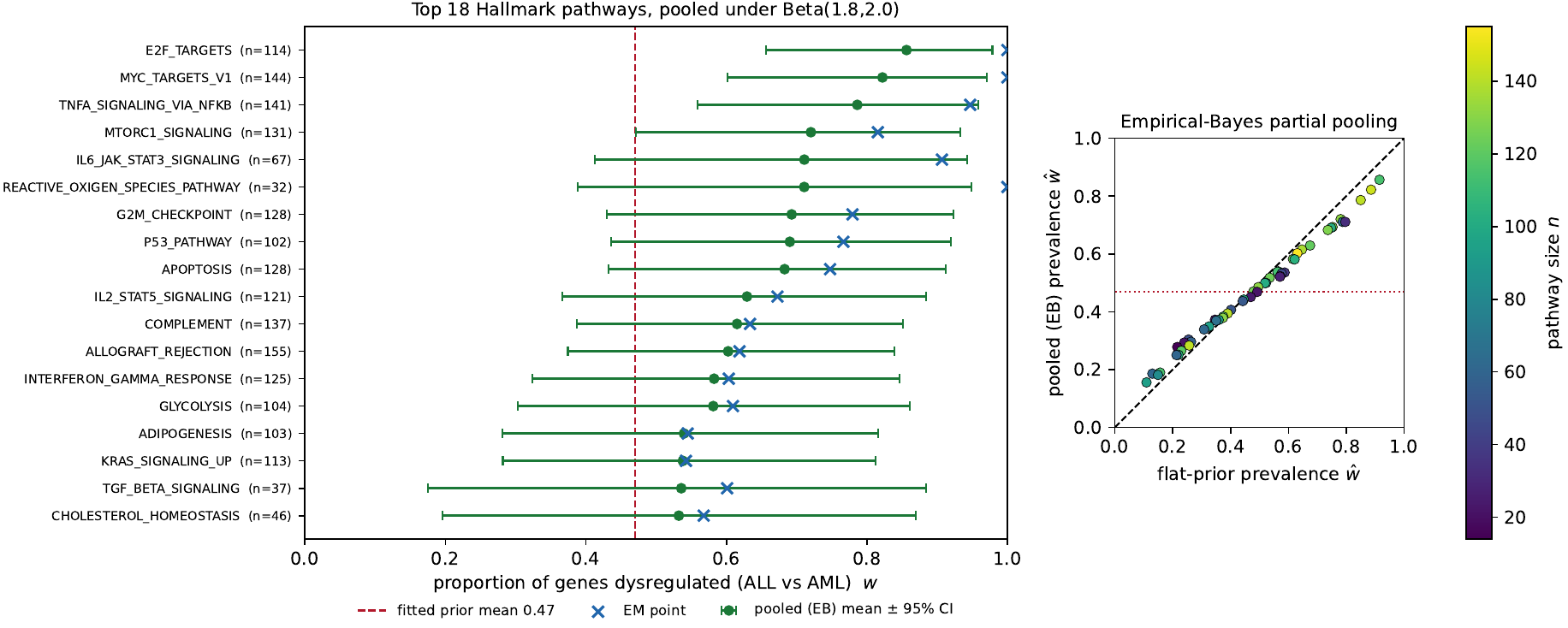}
	\caption{Per-pathway dysregulation prevalence on the Golub leukemia data (ALL vs AML), $50$ MSigDB Hallmark pathways mapped to gene symbols, with an empirical null and a shared prior fitted by empirical Bayes. Left: the $18$ most dysregulated pathways -- partially pooled posterior mean with $95\%$ interval (green) under the evidence-fitted $\mathpzc{Beta}(1.8,2.0)$ prior (mean marked by the dashed line), and the EM point (blue), which saturates at $1$ for the strongest pathways. Right: partial pooling -- each pathway's flat-prior prevalence against its pooled estimate, coloured by pathway size; estimates shrink toward the fitted prior mean, the smallest sets most.}
	\label{fig:app_pathway}
\end{figure}

\subsection{Application III: gene-panel prevalence with known ground truth}\label{sec:app_genomics}

The prevalence intervals of \autoref{sec:app_pathway} are only useful if they are calibrated. Real data rarely comes with a known non-null fraction, so to \emph{measure} calibration on realistic data we build a plasmode from the Golub leukemia CASI expression matrix \citep{Golub1999} ($7128$ genes, $47$ ALL vs $25$ AML). Null genes are drawn as permutation nulls -- a random $12$-vs-$12$ split of the arrays, whose two-sample $z$ (mapped through the $t$ CDF) is exactly $N(0,1)$ -- and signal genes are the $412$ genuinely differential genes ($|z|>4$ on the full cohort) re-evaluated on a $12{+}12$ subsample, so their $z$ is real but realistically attenuated (median $|z|\approx 3$). A ``gene panel'' of size $n$ mixes a known number of signal and null genes; the true prevalence $w$ and every gene's label are known, while the $z$-values are real.

\autoref{fig:app_genomics} reports, over $R=300$ panels per cell, how well each method recovers $w$. The exact $95\%$ interval is calibrated across the grid ($0.95$--$0.98$), whereas the logit-Laplace interval under-covers in the rare-signal corner (down to $0.77$ at $n=20$, $w=0.05$) and the EM$+$Wald interval attains its coverage only by being far too wide (pinned at $1.0$). The single exception is $n=50$, $w=0.20$, where the exact coverage dips to $0.85$: with more signal genes the misspecification of the Gaussian-alternative model for the real, non-Gaussian effect distribution begins to bias the estimate -- a genuine limit, at which the exact method is nonetheless still the best-calibrated of the three. The exact computation is again deterministic and $214\times$ faster than the Gibbs sampler.

\begin{figure}[H]
	\centering
	\includegraphics[width=\textwidth]{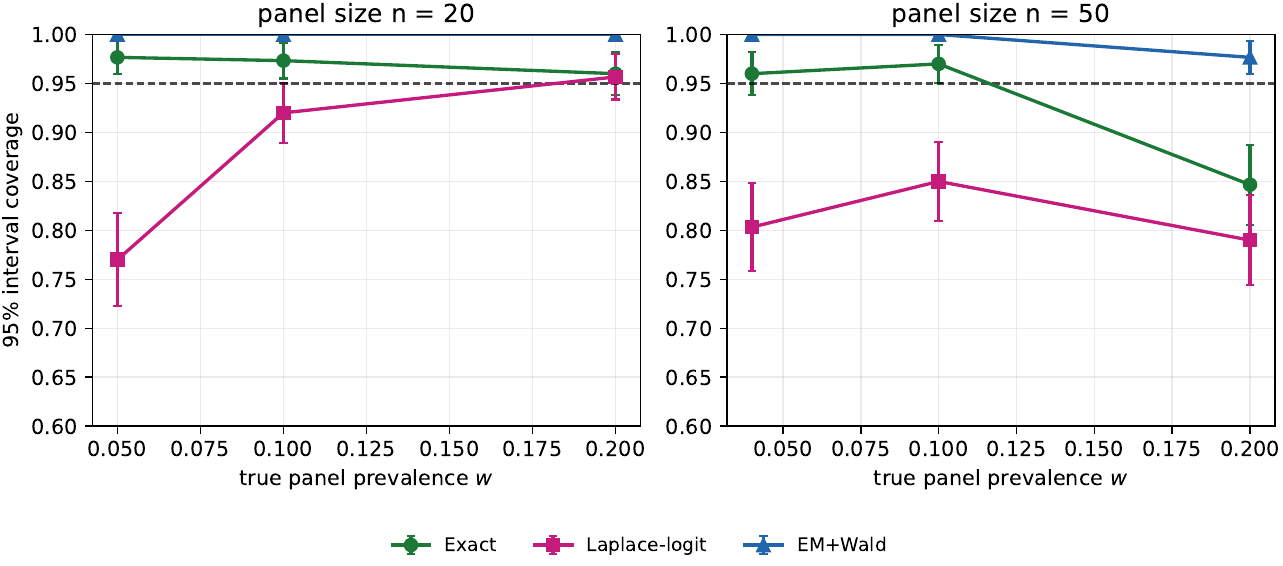}
	\caption{Genomics plasmode from the Golub leukemia data with known ground truth ($R=300$ gene panels per cell). Frequentist coverage of the $95\%$ interval for the true panel prevalence $w$, at panel sizes $n=20$ and $50$. The exact interval is calibrated ($0.95$--$0.98$ across five of the six cells, dipping to $0.85$ in the high-signal $n{=}50$, $w{=}0.20$ corner); logit-Laplace under-covers in the rare-signal corner; EM$+$Wald over-covers with far-too-wide intervals.}
	\label{fig:app_genomics}
\end{figure}

\subsection{Real data: the two-groups model on prostate microarray $z$-values}

We apply the method to the $6033$ gene-level $z$-values of the well-known prostate-cancer microarray study \citep{Singh2002}, a standard testbed for large-scale two-groups inference \citep{Efron2008}. The null component is the theoretical null $g = \mathcal{N}(0,1)$; the non-null component $f = \mathcal{N}(0, \sigma_1^2)$ is a dispersed, symmetric alternative whose scale is chosen by maximising the \emph{exact} marginal evidence (empirical Bayes), giving $\hat\sigma_1 = 1.63$. Here the weight prior is the flat $\mathpzc{Beta}(1,1)$: a single proportion carries no information with which to identify a hyperprior, unlike the multi-stratum meta-analysis and pathway settings above, where the shared prior on the weights is itself fitted by empirical Bayes. The exact posterior of the non-null proportion $w$ has mean $0.178$ with a closed-form $95\%$ interval $[0.148, 0.211]$ -- an interval a Gibbs sampler reproduces to within $0.001$ but that EM cannot provide at all -- computed in milliseconds. The implied $\hat\pi_0 = 0.82$ is specific to this symmetric single-alternative model and is lower than the $\approx\!0.93$ returned by Storey's estimator on the same data; the purpose here is not to re-estimate $\pi_0$ but to show that the entire posterior of the proportion, with calibrated uncertainty, together with the per-gene local FDRs (all $6033$ at once, $59$ genes below $0.2$), issues from a single deterministic pass. \autoref{fig:prostate} shows the fitted mixture, the per-gene local FDR, and the exact-versus-Gibbs posterior of $w$.

\begin{figure}[H]
	\centering
	\includegraphics[width=\textwidth]{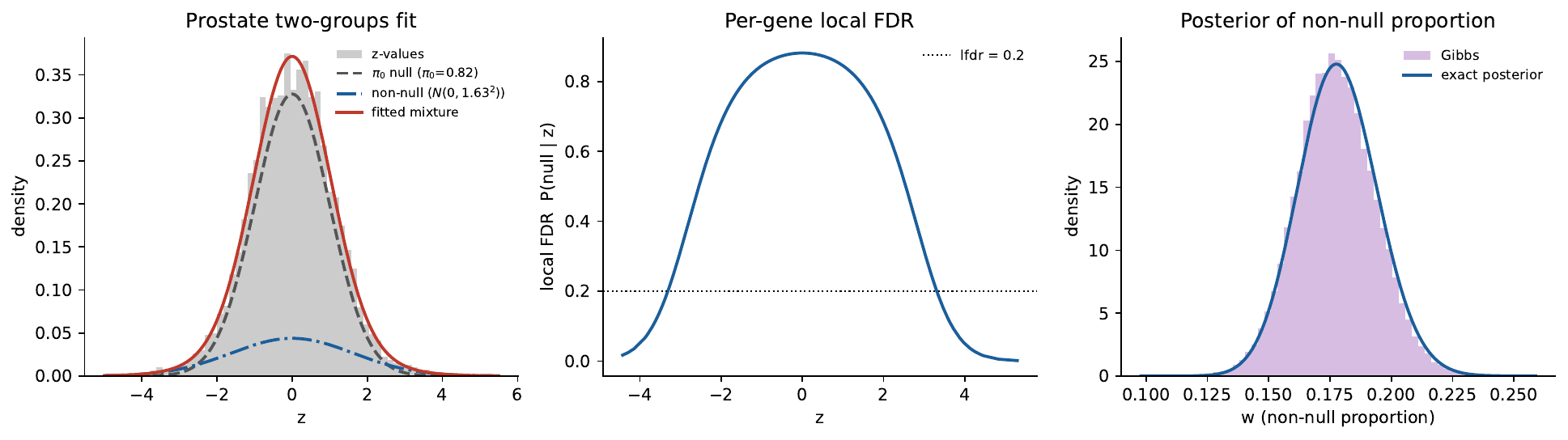}
	\caption{Prostate $z$-values. Left: fitted two-groups mixture (evidence-optimal $\hat\sigma_1 = 1.63$). Middle: exact per-gene local FDR $P(\text{null}\mid z)$. Right: exact posterior of the non-null proportion $w$ against Gibbs samples.}
	\label{fig:prostate}
\end{figure}

To place these numbers against established practice we ran Efron's \texttt{locfdr} \citep{Efron2007} and Strimmer's \texttt{fdrtool} \citep{Strimmer2008} on the same $z$-values (\autoref{fig:locfdr}). On the quantity that drives decisions -- the per-gene local FDR -- our exact method and \texttt{locfdr} under its theoretical null agree almost perfectly: the two local-FDR curves have Spearman rank correlation $0.983$, and their discovery sets at a local-FDR threshold of $0.2$ overlap with Jaccard index $0.915$ ($59$ genes versus $54$). The estimates of the global null proportion, by contrast, spread from $0.82$ to $0.998$ across methods. This spread is expected rather than alarming: $\pi_0$ is not identifiable from the $z$-values alone without an assumption on the alternative near the origin \citep{Efron2008,GenoveseWasserman2004}. Our alternative is a symmetric Gaussian $\mathcal N(0,\sigma_1^2)$ \emph{centred at zero}, so it overlaps the null and softly attributes part of the central bulk to the non-null component, raising the non-null weight to $w=0.178$ and so lowering $\pi_0$ to $0.82$. The other methods instead invoke Efron's ``zero assumption'' -- no non-null density in a neighbourhood of the origin -- which assigns the whole centre to the null and so reports a higher $\pi_0$: Storey's estimator \citep{Storey2002} and \texttt{locfdr}'s theoretical null give $\approx\!0.93$, and \texttt{locfdr}'s empirical null and \texttt{fdrtool}, which fit a \emph{wider} null, give $\approx\!0.99$. The lower $\pi_0$ is thus a different accounting of the centre, not extra signal, and it does not propagate into more discoveries: the model expects $\approx\!1076$ non-null genes, but these are diffuse, high-local-FDR genes that never cross a decision threshold; only the $59$ tail genes below local FDR $0.2$ do (the $25$ below $0.1$ being exactly those \texttt{locfdr} flags), which is why the discovery sets agree despite the gap in $\pi_0$. The spread is a modelling choice shared by every method, not an error of any one; what distinguishes the exact method is that it returns the full posterior of $\pi_0$ with a calibrated $95\%$ interval ($[0.789,0.852]$), whereas \texttt{locfdr} offers only a delta-method standard error and \texttt{fdrtool}/Storey a point estimate. We stress that on a single dataset of $6033$ genes every method has ample data, so the small-sample calibration advantage of \autoref{sec:calibration} is not what separates them here; the exact method's contribution on this benchmark is that its gene-level decisions match the standard tool while its global summary comes with deterministic, quantified uncertainty.

\begin{figure}[H]
	\centering
	\includegraphics[width=\textwidth]{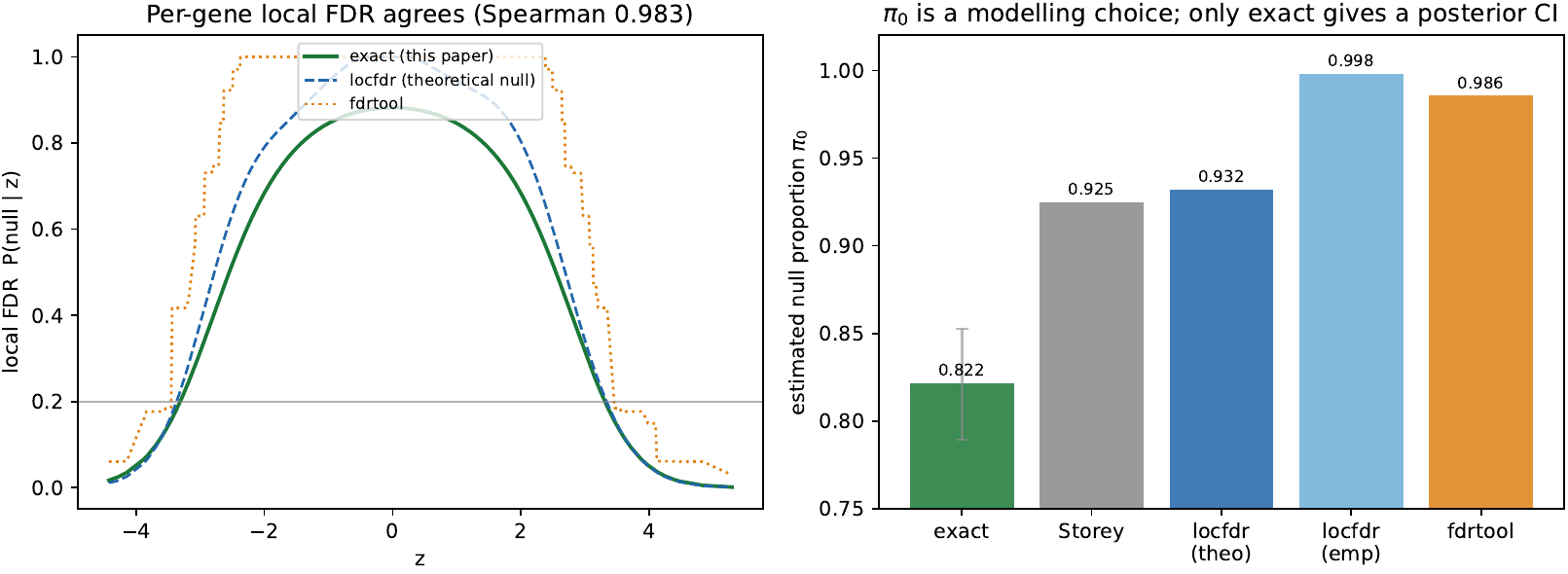}
	\caption{Exact two-groups model versus \texttt{locfdr} and \texttt{fdrtool} on the prostate $z$-values. Left: the per-gene local FDR of the exact method coincides with \texttt{locfdr}'s under a theoretical null (Spearman $0.983$); \texttt{fdrtool}'s spline estimate is more jagged. Right: estimated null proportion $\pi_0$ across methods -- it ranges from $0.82$ to $0.998$ depending on the null/alternative model, and only the exact method attaches a posterior credible interval (error bar).}
	\label{fig:locfdr}
\end{figure}

\subsection{Hierarchy: partial pooling across studies}

The model of \autoref{fig:model} places a shared prior over many strata. We simulate a meta-analysis of $M = 150$ studies, each on only $8$ to $119$ observations -- the small-stratum regime where borrowing strength across studies actually matters -- with study-specific non-null proportions $w_m \sim \mathpzc{Beta}(2,6)$, and fit the shared prior by maximising the \emph{total} exact evidence, a smooth deterministic objective solved in $163$ evaluations in about $6$~s with no inner MCMC or EM loop. The fitted prior recovers the truth (mean $0.257$ versus $0.25$), and the resulting partially pooled posterior means reduce the root-mean-square error over the unpooled per-study maximum-likelihood estimates by $29\%$ overall -- and by $35\%$ for the smallest third of the studies against $12\%$ for the largest, shrinkage buying the most where the per-study data are scarcest -- with $98\%$ credible-interval coverage of the true $w_m$ (\autoref{fig:hierarchical}). This comparison mixes two effects, the informative fitted prior and the pooling of information across studies, and the exact percentage varies with the simulation; what matters here is that empirical-Bayes pooling is a single deterministic evidence maximisation, not a nested sampling problem.

\begin{figure}[H]
	\centering
	\includegraphics[width=\textwidth]{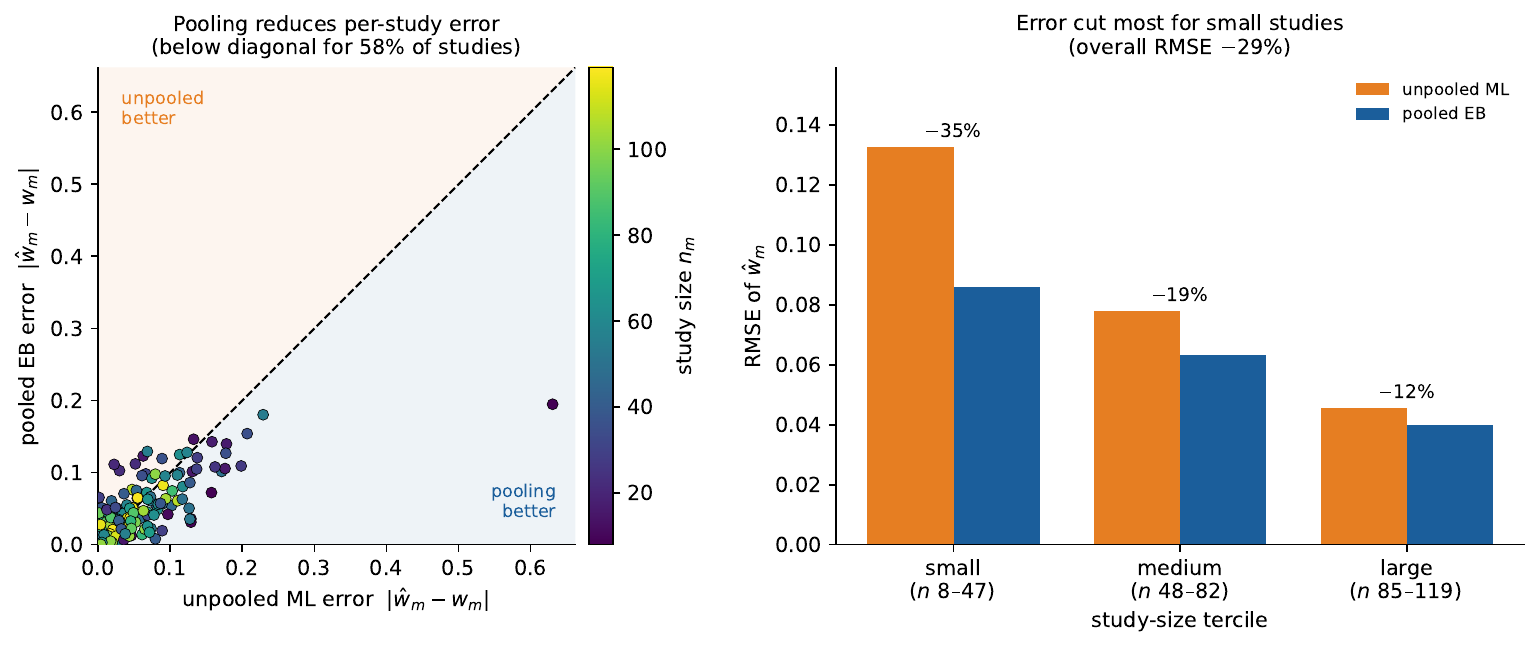}
	\caption{Empirical-Bayes partial pooling across $150$ small studies. Left: per-study absolute error of the pooled empirical-Bayes mean against the unpooled maximum-likelihood estimate, coloured by study size -- points below the diagonal (here $58\%$ of studies, and every study with a large unpooled error) are ones pooling improves. Right: root-mean-square error by study-size tercile -- pooling cuts the error by $35\%$ for the smallest third of studies and $29\%$ overall, the gain shrinking as the per-study sample grows.}
	\label{fig:hierarchical}
\end{figure}

\subsection{The empirical-Bayes objective is deterministic}\label{sec:ebdet}

The evidence surface that drives the empirical-Bayes fits above is not only cheap but \emph{noise-free}, which matters when it is optimised. \autoref{fig:ebdet} fits the alternative scale $\sigma_1$ by maximising the marginal evidence on a single two-groups dataset ($n=300$, $w=0.2$). Exact marginalisation returns a smooth, unimodal objective and a single reproducible maximiser $\hat\sigma_1 = 1.876$ over the whole $60$-point grid in $75$~ms. An importance-sampling Monte-Carlo estimate of the \emph{same} evidence is unbiased but noisy near the flat optimum, so its argmax is itself a random variable: over $25$ seeds it scatters with standard deviation $0.8$ grid steps at $M=2000$ draws ($1.5$ grid steps at $M=200$), at roughly $290\times$ the cost. The gain is thus one of reproducibility and speed of hyperparameter selection, not of point-estimate bias -- a Monte-Carlo evidence would, with enough draws, find the same optimum on average.

\begin{figure}[H]
	\centering
	\includegraphics[width=\textwidth]{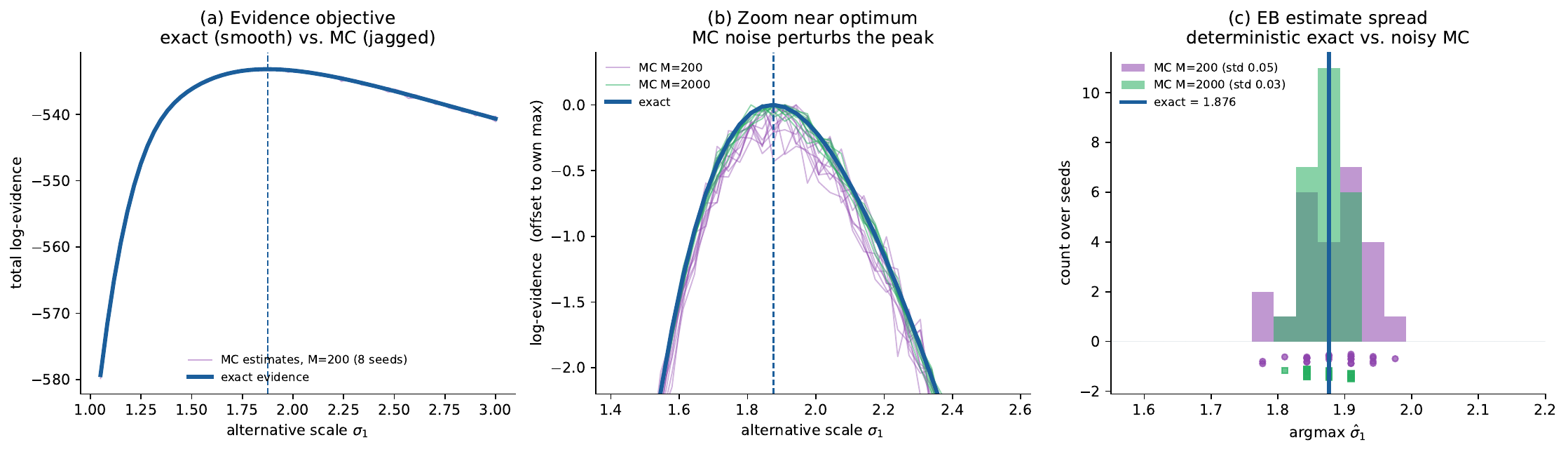}
	\caption{The exact evidence gives a deterministic empirical-Bayes objective. Left: the exact log-evidence as a function of $\sigma_1$ (smooth, single maximiser $\hat\sigma_1=1.876$). Right: importance-sampling Monte-Carlo estimates of the same evidence, offset to their own maxima; the estimated optimum scatters across seeds, so Monte-Carlo empirical Bayes carries a selection variance that exact marginalisation removes.}
	\label{fig:ebdet}
\end{figure}

\subsection{The three-group model on real data}

Finally we exercise the three-group model (down-regulated, null, up-regulated) with a Beta-Liouville prior on a block of the prostate data (\autoref{fig:kmixture}). The purpose here is to validate the $K$-component algorithm on real data, not to claim a model selection: the exact log-domain program \eqref{eq:kmix_final} agrees with Monte-Carlo integration over the simplex to $0.02\%$. The separable product $\prod_j e_j(c_j)$ of Remark~\ref{rem:correction}, which ignores the disjointness constraint, returns a materially different value -- on well-separated synthetic data the discrepancy exceeds half a log-unit -- confirming numerically that the joint dynamic program is needed for $K \ge 3$. Because the exact marginal is a genuine model evidence it \emph{could} drive Bayesian model comparison; we deliberately report no two- versus three-group Bayes factor, because a fair comparison must grant each model its own evidence-optimal hyperparameters and prior, and under that constraint the ranking is sensitive to those choices. A calibrated model-comparison study is left to future work.

\begin{figure}[H]
	\centering
	\includegraphics[width=\textwidth]{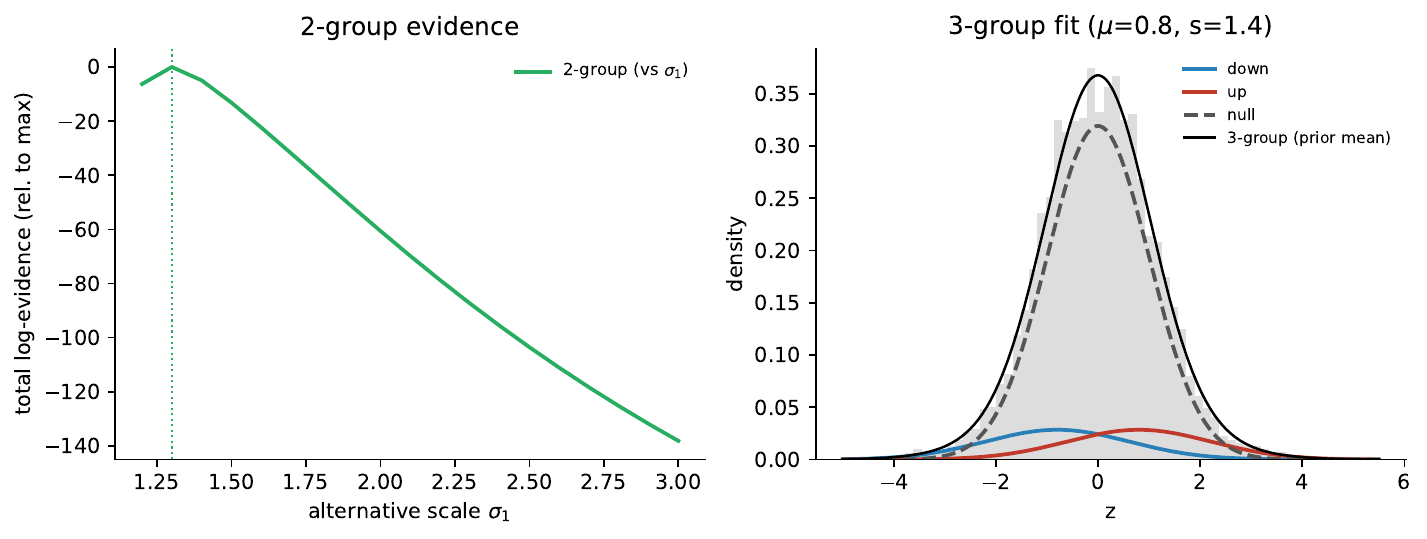}
	\caption{Three-group Beta-Liouville model on the prostate data. Left: total evidence of the two-group model as its alternative scale varies. Right: the evidence-fitted three-group density (down / null / up components).}
	\label{fig:kmixture}
\end{figure}

\section{Conclusion}\label{sec:conclusion}

We have developed exact, deterministic algorithms for marginalising the mixing weight of a hierarchical two- and $K$-component mixture under Beta and Beta-Liouville priors. For the two-component model the marginal likelihood is a Beta-weighted sum of elementary symmetric polynomials computable by an $O(n^2)$ log-domain dynamic program (with an $O(n\log^2 n)$ greedy-FFT scheme available as a complexity result, though numerically usable only for small $n$), and, as a corollary, the exact posterior of the weight is a finite mixture of Beta distributions -- giving closed-form posterior means, variances, credible intervals, and per-observation posterior class probabilities (local FDR) without sampling. This closed form is more than convenient: the posterior over the latent signal count is log-concave and unimodal by Newton's inequalities (\autoref{thm:logconcave}), its moments are exact ratios of the marginal likelihood at shifted hyperparameters (\autoref{prop:moments}), and it responds monotonically to the data (\autoref{prop:monotone}) -- guarantees the boundary-collapsing approximations it replaces do not carry. On the prostate-cancer data the method reproduces MCMC posteriors to three or four significant figures at a fraction of the cost for small-to-moderate $n$, and delivers calibrated credible intervals for the mixing proportion in the small-sample, rare-signal regime where cheap Gaussian/Laplace approximations are not (\autoref{sec:calibration}), through a smooth, deterministic empirical-Bayes objective that scales across the strata of a hierarchical model. This is where the method helps most, in the applications it is built for (\autoref{sec:app_meta}--\ref{sec:app_genomics}). In a real multilevel meta-analysis EM collapses to boundary point estimates for most small strata. In a pathway-level readout of leukemia expression the exact posterior returns a calibrated dysregulation prevalence per Hallmark pathway that enrichment tests and EM do not. And on a leukemia-derived gene-panel benchmark with known ground truth the exact interval alone is calibrated. Where instead the data are abundant (the $6033$-gene prostate benchmark) the method matches standard tools such as \texttt{locfdr} on the gene ranking and merely adds a posterior interval for the null proportion.

For $K \ge 3$ the coefficient in the multivariate expansion is a disjoint-support symmetric sum, which does not factor into a product of univariate elementary symmetric polynomials; the exact marginal is then obtained by a joint count dynamic program of cost $O(n^{K-1})$ in storage -- very practical for the small $K$ that arise in applications, particularly within a hierarchy of moderate strata. Natural extensions include estimating the component parameters $\theta$ jointly with the weights by maximising the exact evidence, and exploiting the exact per-item posteriors for calibrated large-scale decision rules.

\subsection*{Reproducibility}\label{sec:reproducibility}

A reference Python implementation of every algorithm, benchmark, and experiment in this paper -- together with the validation suite against brute-force, quadrature, arbitrary-precision, and Monte-Carlo references -- accompanies this manuscript.

\subsection*{Acknowledgements}

\newpage
\appendix

\section{Beta-Liouville distribution}

\subsection{Definition: stochastic representation}
Let there be 2 independent r.vs.: $\mathbf{d} \sim \mathpzc{Dir}_{K-2}(\mathbf{\boldsymbol{\hat{\alpha}}})$ (where $\boldsymbol{\hat{\alpha}} = (\alpha_1, \dots, \alpha_{K-1})$) and $\beta \sim \mathpzc{Beta}(\gamma, \alpha_{K})$. Then, as introduced by \cite{Gupta1996}, r.v. $\bold{w}$ with components

\begin{equation}\label{eq:stochastic_representation}
w_i = \beta d_i,~i = 1, \dots, K-2, ~~ w_{K-1} = \beta (1 - \sum_{i=1}^{K-2} d_i)
\end{equation}
follows the Beta-Liouville distribution.
$$
\bold{w} \sim \mathpzc{BL}_{K-1}(\boldsymbol{\alpha}, \gamma), ~\boldsymbol{\alpha} = (\alpha_1, \dots, \alpha_{K-1})
$$
Note that $\sum_{i=1}^{K-1} w_i < 1$ and to complete the simplex we have to introduce $w_K = 1 - \sum_{i=1}^{K-1} w_{i}$ explicitly. See that $w_K = 1 - \beta$. It follows that for $K=2$ $\mathpzc{BL}$ is just a $\mathpzc{Beta}$ distribution.

\subsection{Probability density function}\label{app:bl_pdf}
As $\bold{d}$ and $\beta$ are independent, their joint PDF is simply product for the Dirichlet PDF and Beta PDF:

$$
g(\bold{d}, \beta) = \underbrace{\frac{1}{B(\boldsymbol{\hat{\alpha}})} \left(1-\sum_{i=1}^{K-2}d_i\right)^{\alpha_{K-1} - 1} \prod_{i=1}^{K-2} d_i^{\alpha_i - 1}}_{\mathpzc{Dir_{K-2}}} \underbrace{\frac{1}{B(\gamma,\alpha_K)} \beta^{\gamma - 1} (1 - \beta)^{\alpha_K - 1}}_{\mathpzc{Beta}},
$$
where $B(\gamma, \alpha_K)$ is the Beta function and $B(\boldsymbol{\hat{\alpha}})$ is the multivariate Beta function.
See that number of r.v.s is unchanged. The inverse of the transformation \autoref{eq:stochastic_representation} is:

$$
\beta = \sum_{i=1}^{K-1} w_i, \hspace{0.5em} d_i = \frac{w_i}{\beta} =  \frac{w_i}{\sum_{i=1}^{K-1} w_i}, 
$$

The determinant of Jacobian of the transformation:
\begin{align*}
|J_{(\bold{d}, \beta) \rightarrow \bold{w}}| = \left| \begin{bmatrix}
	\frac{\partial w_1}{\partial d_1} & \frac{\partial w_1}{\partial d_2} & \cdots & \frac{\partial w_1}{\partial d_{K-2}} & \frac{\partial w_1}{\partial \beta} \\
	\frac{\partial w_2}{\partial d_1} & \frac{\partial w_2}{\partial d_2} & \cdots & \frac{\partial w_2}{\partial d_{K-2}} & \frac{\partial w_2}{\partial \beta} \\
	\vdots & \vdots & \ddots & \vdots & \vdots \\
	\frac{\partial w_{K-1}}{\partial d_1} & \frac{\partial w_{K-1}}{\partial d_2} & \cdots & \frac{\partial w_{K-1}}{\partial d_{K-2}} & \frac{\partial w_{K-1}}{\partial \beta} 
\end{bmatrix}\right| = & \left| \begin{bmatrix}
\beta & 0 & 0 & \cdots & 0 & d_1 \\
0 & \beta & 0 & \cdots & 0 & d_2 \\
0 & 0 & \beta & \cdots & 0 & d_3 \\
\vdots & \vdots & \vdots & \ddots & \vdots & \vdots \\
0 & 0 & 0 & \cdots & \beta & d_{K-2} \\
- \beta & - \beta & - \beta & \cdots & - \beta & 1 - \sum_{i=1}^{K-1} d_i
\end{bmatrix} \right| = \\
= &
\left| \begin{bmatrix}
	\beta & 0 & 0 & \cdots & 0 & d_1 \\
	0 & \beta & 0 & \cdots & 0 & d_2 \\
	0 & 0 & \beta & \cdots & 0 & d_3 \\
	\vdots & \vdots & \vdots & \ddots & \vdots & \vdots \\
	0 & 0 & 0 & \cdots & \beta & d_{K-2} \\
	0 & 0 & 0 & \cdots & 0 & 1
\end{bmatrix} \right|  = \beta^{K-2}.
\end{align*}

Hence, 

\begin{align}
	f(\bold{w}) = g\left(\frac{\bold{w}}{\beta}, \sum_{i=1}^{K-1}w_i\right) |J_{(\bold{d}, \beta) \rightarrow \bold{w}}| =  \frac{1}{B(\boldsymbol{{\hat{\alpha}}}) B(\alpha_{K}, \gamma)}  \left(\sum_{i=1}^{K-1} w_i \right) ^{\gamma - \sum_{i=1}^{K-1} \alpha_i} \left(\underbrace{1 - \sum_{i=1}^{K-1} w_i}_{w_{K}} \right) ^{\alpha_{K} - 1} \prod_{i=1}^{K-1} w_i^{\alpha_i-1} = \\ =  \frac{1}{B(\boldsymbol{{\hat{\alpha}}}) B(\alpha_{K}, \gamma)}  \left(\sum_{i=1}^{K-1} w_i \right) ^{\gamma - \sum_{i=1}^{K-1} \alpha_i} \prod_{i=1}^{K} w_i^{\alpha_i-1}.
\end{align}

Note that for $\gamma = \sum_{i=1}^{K-1} \alpha_i$ $\mathpzc{BL}$ agrees with $\mathpzc{Dir}$ distribution. 

\subsection{Moments}\label{app:bl_moments}

We want to find a mixed moment
$\EX\left[\prod_{i=1}^K w_i^{c_i}\right]$, where $\bold{c} \in \Re^{K}$. To this end, let's use the stochastic representation of $\mathpzc{BL}$ from \autoref{eq:stochastic_representation}:

\begin{align*}
\EX\left[\prod_{i=1}^K w_i^{c_i}\right] = \EX\left [w_K^{c_K} w_{K-1}^{c_{K-1}} \prod_{i=1}^{K-2} w_i^{c_i} \right] = 
\EX\left [\left(1 - \beta \right)^{c_K} \left(\beta \left(1 - \sum_{i=1}^{K-2} d_i\right)\right)^{c_{K-1}} \prod_{i=1}^{K-2} (\beta d_i)^{c_i} \right] = \\
=
\EX\left [\beta^{\sum_{i=1}^{K-1} c_i}\left(1 - \beta \right)^{c_K}  \prod_{i=1}^{K-1} d_i^{c_i} \right] = \EX\left[\beta^{\sum_{i=1}^{K-1} c_i}\left(1 - \beta \right)^{c_K}  \right] \EX\left[\prod_{i=1}^{K-1} d_i^{c_i} \right].
\end{align*}

This a product of 2 independent r.v.s, namely $\beta \sim \mathpzc{Beta}(\gamma, \alpha_K) = \mathpzc{Dir}(\gamma, \alpha_K)$ and $d_i \sim \mathpzc{Dir}(\boldsymbol{\hat{\alpha}})$. The mixed distribution of the $\mathpzc{Dirichlet}$ r.v. is well-known:

$$
\EX\left[\prod_{i=1}^{K-1} d_i^{c_i} \right] = \frac{B(\boldsymbol{\hat{\alpha}} + \mathbf{\hat{c}})}{B(\boldsymbol{\hat{\alpha}})}, \hspace{0.5em} \EX\left[\beta^{\sum_{i=1}^{K-1} c_i}\left(1 - \beta \right)^{c_K}\right] = \frac{B(\gamma + \sum_{i=1}^{K-1} c_i, \alpha_K + c_K)}{B(\gamma, \alpha_K)},
$$
hence

\begin{equation}\label{eq:bl_mixed_moments}
	\EX\left[\prod_{i=1}^K w_i^{c_i}\right] = \frac{B(\boldsymbol{\hat{\alpha}} + \mathbf{\hat{c}})}{B(\boldsymbol{\hat{\alpha}})} \frac{B(\gamma + \sum_{i=1}^{K-1} c_i, \alpha_K + c_K)}{B(\gamma, \alpha_K)}.
\end{equation}

\subsection{Covariance}

Given the \autoref{eq:bl_mixed_moments}, it is trivial to compute covariance between $w_i$ and $w_j$:
\begin{equation*}
	cov(w_i, w_j) = \EX[w_i w_j] - \EX[w_i]\EX[w_j] 
\end{equation*}

The exact formula will differ for varying $i$ and $j$:
\begin{itemize}
	\item If both $i$ and $j$ less thatn $K$:
	\begin{equation}
		cov(w_i, w_j) = \frac{\alpha_i \alpha_j \gamma \left( (\sum_{k=1}^{K-1} \alpha_k) \alpha_K - \gamma (\gamma + \alpha_K + 1) \right)}{(\sum_{k=1}^{K-1} \alpha_k)^2 (\gamma + \alpha_K)^2 (\gamma + \alpha_K + 1)(\sum_{k=1}^{K-1} \alpha_k + 1)}
	\end{equation}
	\item If $i < K$ and $j = K$:
	\begin{equation}
	cov(w_i, w_K) =	\frac{-\alpha_i \gamma \alpha_K}{(\sum_{k=1}^{K-1} \alpha_k)(\gamma + \alpha_K)^2 (\gamma + \alpha_K + 1)}
	\end{equation}
\end{itemize}

\subsection{Interpretation}
Unlike the $\mathpzc{Dir}$ distribution ($\gamma = \sum_{i=1}^{K-1} \alpha_i$), $\mathpzc{BL}$ introduces a dependence-modifying term $\left(\sum w_i\right)^{\gamma - \sum \alpha_i}$, controlled by $\gamma$. This allows modeling of both positive and negative associations for $w_i, w_j$, where $i < K$ and $j < K$. Note that covariances between the last component and other components is always negative for both $\mathpzc{BL}$ and $\mathpzc{Dir}$.

\subsubsection*{Beta-Liouville as a stick-breaking process}

Alternatively, $\mathpzc{BL}$ under a certain parametrization can be viewed as a stick-breaking process. For a unit stick, define $\beta_i \sim \text{Beta}(a_i, b_i)$. The stick-breaking weights are:
\[
\pi_i = \beta_i \prod_{j=1}^{i-1} (1 - \beta_j), \quad \sum_{i=1}^K \pi_i = 1.
\]

Let's construct $\mathpzc{BL}$ as:
\begin{enumerate}
	\item Draw global scale: $Y \sim \text{Beta}(a, b)$;
	\item Allocate $Y$ via Dirichlet: $\mathbf{V} \sim \text{Dirichlet}(\alpha_1, \dots, \alpha_{K-1})$;
	\item Assign weights:
	$$
	w_i = Y V_i \quad (i < K), \quad w_K = 1 - Y \sum_{i=1}^{K-1} V_i.
	$$
\end{enumerate}

Because $\mathbf{V}$ is a full Dirichlet vector on the $(K-1)$-simplex, $\sum_{i=1}^{K-1} V_i = 1$ and hence $w_K = 1 - Y$. Setting $a = \gamma$ and $b = \alpha_K$, so that $Y \sim \mathpzc{Beta}(\gamma, \alpha_K)$, the construction reproduces exactly the stochastic representation of \autoref{eq:stochastic_representation} (with $Y = \beta$ and $\mathbf{V} = \mathbf{d}$) and therefore attains $\mathpzc{BL}$ form. (In particular $Y$ plays the role of the canonical scale $\beta$; this is not a distinct stick-breaking recursion but a restatement of the stochastic representation.)

\newpage
\section{Proofs of the structural properties}\label{app:properties_proofs}

\begin{proof}[Proof of \autoref{lem:newton} (Newton's inequalities)]
	Since every $t_i > 0$, the polynomial $p(z) = \prod_{i=1}^n (1 + t_i z) = \sum_{m=0}^n e_m(\mathbf{t})\, z^m$ has only real (negative) roots. Two operations preserve real-rootedness of a polynomial with only real roots: differentiation, by Rolle's theorem, and coefficient reversal $p(z) \mapsto z^{\deg p}\, p(1/z)$. A composition of $m-1$ differentiations, one reversal, and $n-m-1$ further differentiations carries $p$, up to a positive scalar, to the quadratic
	\begin{equation*}
		E_{m-1} + 2\,E_m\, z + E_{m+1}\, z^2, \qquad E_k = e_k(\mathbf{t})/\tbinom{n}{k},
	\end{equation*}
	(the binomial denominators arise precisely from the repeated differentiation of $z^k$). This quadratic still has only real roots, so its discriminant is nonnegative, $(2E_m)^2 \ge 4\,E_{m-1}E_{m+1}$, which is \eqref{eq:newton}. Equality at every $m$ forces a double root throughout, i.e.\ all $t_i$ equal; see \citet[Ch.~2]{hardy1952inequalities} for the classical account.
\end{proof}

\begin{proof}[Log-concavity of the Beta--Binomial factor in \autoref{thm:logconcave}]
	Log-concavity of $\{b_m\}$ is equivalent to the ratio $r(m) := b_{m+1}/b_m$ of \eqref{eq:bb_ratio} being non-increasing in $m$. Let $\varphi(x) := \log\!\big(x/(x-1)\big)$, which is decreasing on $(1,\infty)$ because $x/(x-1) = 1 + (x-1)^{-1}$ is. From \eqref{eq:bb_ratio},
	\begin{equation*}
		\log r(m) - \log r(m+1) = \underbrace{\big[\varphi(n-m) - \varphi(n-m-1+\beta)\big]}_{\text{needs } \beta \ge 1} + \underbrace{\big[\varphi(m+2) - \varphi(m+1+\alpha)\big]}_{\text{needs } \alpha \ge 1}.
	\end{equation*}
	When $\beta \ge 1$ we have $n-m-1+\beta \ge n-m$, so the first bracket is $\ge 0$ because $\varphi$ is decreasing; when $\alpha \ge 1$ we have $m+1+\alpha \ge m+2$, so the second bracket is likewise $\ge 0$. Hence $r(m) \ge r(m+1)$, and $\{b_m\}$ is log-concave. Each bracket also shows the roles of the two hypotheses are separate: violating $\alpha \ge 1$ or $\beta \ge 1$ can make the corresponding bracket negative, and log-concavity then fails.
\end{proof}

\section{Proposition evidences}
This appendix contains some easy-to-run scripts for those in doubt of some of combinatorial identities derived in the main part of the paper.

\subsection*{Proposition on coronation and deposition}
To see that the proposition holds true for some values of $n, d, i$, run the program:

\begin{lstlisting}[language=Python,style=code]
from sympy.functions.combinatorial.numbers import stirling, binomial

n = 100
for i in range(n + 1 + 1):
	for d in range(n + 1 - i  + 1):
		s1 = 0
		s2 = 0
		for j in range(d, n + 1 - i + 1):
			s1 += stirling(j + i  - d, i, kind=2) * \
			      stirling(n + 1, j + i, kind=1) *\
			      binomial(j + i, d)
			s2 += stirling(n + 1 - i, j, kind=1) *\
			      binomial(j, d) * \
			      i **(j - d)
		s2 *= binomial(n + 1, i)
		if s1 != s2:
			raise Exception("Oh no! The conjecture is false!")
\end{lstlisting}

\subsection*{Proposition on the values of $c_i$ coefficients}
To see that proposition holds true for some values of $n, i$, run the program:
\begin{lstlisting}[language=Python,style=code]
import sympy as sp
from sympy.functions.combinatorial.numbers import stirling, binomial

def calc_lhs(beta, n, i):
    total = 0
    for j in range(i, n + 2):  
        sign = (-1) ** (i + n + 1)
        s1 = stirling(j, i, kind=2)
        inner_sum = 0
        for k in range(j, n + 2):  
            s2 = stirling(n + 1, k, kind=1)
            b = binomial(k, j)
            term = s2 * b * (beta ** (k - j))
            inner_sum += term
        total += sign * s1 * inner_sum
    return total

def calc_rhs(beta, n, i):
    exponent = n + 1 - i
    sign = (-1) ** exponent
    rf = sp.RisingFactorial(beta + i, exponent)
    binom = binomial(n + 1, i)
    return sign * rf * binom


beta = sp.symbols('beta')
for n in range(1, 20):
    for i in range(0, n):
        lhs = calc_lhs(beta, n, i)
        rhs = calc_rhs(beta, n, i)
        diff = rhs - lhs
        if diff.simplify() != 0:
            raise Exception('On, no! The conjecture is false!')
\end{lstlisting}

\section{Polynomial multiplication pseudocode}\label{app:polymult_code}
\begin{algorithm}
	\caption{PolyMult: Multiply two polynomials}
	\begin{algorithmic}[1]
		\Procedure{PolyMult}{$A, B, m$}
		\State \textbf{Input:} Polynomials $A$ and $B$; threshold degree $m$
		\State \textbf{Output:} Polynomial $C(z) = A(z) \cdot B(z)$
		\State $d_A \gets \deg(A)$
		\State $d_B \gets \deg(B)$
		\State $d_{\min} \gets \min\{d_A, d_B\}$
		\If{$d_{\min} > m$}
		\State \Return \Call{FFTMultiply}{$A, B$}
		\Else
		\State \Return \Call{OrdinaryMultiply}{$A, B$}
		\EndIf
		\EndProcedure
	\end{algorithmic}
	\begin{algorithmic}[1]
		\Procedure{FFTMultiply}{$A, B$}
		\State \textbf{Input:} Polynomials $A$ and $B$
		\State \textbf{Output:} Polynomial $C(z) = A(z) \cdot B(z)$
		\State $n_A \gets \deg(A) + 1$
		\State $n_B \gets \deg(B) + 1$
		\State $n \gets \text{smallest power of }2 \text{ such that } n \geq n_A + n_B - 1$
		\State \textbf{Pad} $A$ and $B$ with zeros so that their lengths equal $n$
		\State $A \gets \textsc{FFT}(A, n)$ \Comment{Compute the FFT of $A$}
		\State $B \gets \textsc{FFT}(B, n)$ \Comment{Compute the FFT of $B$}
		\For{$i \gets 0$ to $n-1$}
		\State $C[i] \gets A[i] \times B[i]$
		\EndFor
		\State $C \gets \textsc{InverseFFT}(C, n)$ \Comment{Compute the inverse FFT to get coefficients}
		\State \textbf{return} $C[0 \ldots (n_A+n_B-2)]$ \Comment{Truncate to the correct degree}
		\EndProcedure
	\end{algorithmic}
	\begin{algorithmic}[1]
		\Procedure{OrdinaryMultiply}{$A, B$}
		\State \textbf{Input:} Polynomials $A$ and $B$
		\State \textbf{Output:} Polynomial $C(z) = A(z) \cdot B(z)$
		\State $n_A \gets \deg(A)$
		\State $n_B \gets \deg(B)$
		\State Initialize $C$ as an array of zeros of length $(n_A + n_B + 1)$
		\For{$i \gets 0$ to $n_A$}
		\For{$j \gets 0$ to $n_B$}
		\State $C[i+j] \gets C[i+j] + A[i] \times B[j]$
		\EndFor
		\EndFor
		\State \textbf{return} $C$
		\EndProcedure
	\end{algorithmic}
\end{algorithm}

\newpage
\section{The greedy-FFT variant: complexity and numerical failure}\label{app:sect_fft_asstmptotic}

The greedy-FFT pairing computes the elementary symmetric polynomials with the best asymptotic complexity of any backend here, $O(n\log^2 n)$; we prove this below. It is nonetheless \emph{not} a usable way to evaluate the marginal in double precision, and we recommend the $O(n^2)$ log-domain DP throughout. This appendix records both facts.

\subsection*{Complexity}
\begin{proposition}
	Given $n$ linear polynomials, the greedy FFT pairing algorithm computes their product in $O(n \log^2 n)$ time.
\end{proposition}

\begin{proof}
	Suppose we are given $n$ polynomials
	$$
	P_i(z) = 1 + t_i z, \quad i=1,\dots,n.
	$$
	The goal is to compute
	$$
	P(z) = \prod_{i=1}^n (1+t_i z)
	$$
	by repeatedly merging polynomials using FFT-based multiplication when beneficial. We build a binary merge tree where each leaf is one of the linear factors and each internal node corresponds to the product of its two children.
	
	At a generic merge step, assume two polynomials $A(z)$ and $B(z)$ of degrees $d_A$ and $d_B$ are multiplied. The FFT-based multiplication computes the product in
	$$
	O\bigl((d_A+d_B) \log (d_A+d_B)\bigr)
	$$
	time.
	
	By merging the smallest-degree polynomials first (the \emph{greedy} strategy), we ensure that the degrees of the polynomials in intermediate steps grow as slowly as possible.
	
	Let $d_v$ denote the degree of the polynomial computed at an internal node $v$ of the merge tree; for its two children $A$ and $B$ we have $d_v = d_A + d_B$. The greedy strategy merges the two lowest-degree polynomials at each step, which keeps the tree balanced with depth $L = O(\log n)$. Group the internal nodes by tree level. The polynomials residing at a fixed level correspond to disjoint groups of the original $n$ leaves, so the sum of their degrees at that level is at most $n$. Summing over the $L = O(\log n)$ levels therefore gives
	$$
	\sum_{v \in \text{internal nodes}} d_v = O(n \log n).
	$$
	(The often-quoted figure $2n-1$ is the \emph{number} of nodes and the final degree accounting, \emph{not} the sum of intermediate degrees; each of the $\log_2 n$ levels contributes total degree up to $n$, and one verifies numerically that $\sum_v d_v = n\log_2 n$ exactly for $n$ a power of two.)

	Thus, the total time $T(n)$ is bounded by
	$$
	T(n) = \sum_{\text{internal nodes } v} O\Bigl( d_v \log (d_v) \Bigr).
	$$
	Since the maximum degree encountered is at most $n$, we have $\log(d_v) \le \log n$ for all internal nodes. Consequently,
	$$
	T(n) \le \log n \sum_{v} O(d_v) = O\bigl(\log n \cdot n \log n\bigr) = O(n \log^2 n).
	$$
\end{proof}

\subsection*{Why it is not numerically usable}

The FFT necessarily works in the linear domain, where the elementary symmetric polynomials span a combinatorial range ($e_m \sim \binom{n}{m}$, amplified further for well-separated components). The error of an FFT convolution scales with the \emph{largest} coefficient, so the small $e_m$ -- which dominate $\log Z$ when the mixing weight is small -- lose all significance, and round-off can even drive them negative. \autoref{fig:accuracy} scores the three double-precision backends against a $60$-digit exact-arithmetic reference. The log-domain DP holds $\sim\!10^{-12}$ relative error for all $n$ and all separations. The plain-domain DP stays accurate until it cleanly overflows to \texttt{NaN} near $n\approx10^3$. The FFT, however, returns \emph{finite but grossly wrong} values well before overflow -- already off by $\sim\!200$ nats at $n=200$, and by $\sim\!100$ nats even at zero component separation. Because the corruption is silent (a plausible finite number, not an obvious \texttt{NaN}), the FFT backend is dangerous to use for the marginal; our implementation now surfaces the corrupted coefficients as \texttt{NaN} rather than sanitising them. We therefore keep the greedy FFT for its complexity interest only and default to the log-domain DP everywhere.

\begin{figure}[H]
	\centering
	\includegraphics[width=\textwidth]{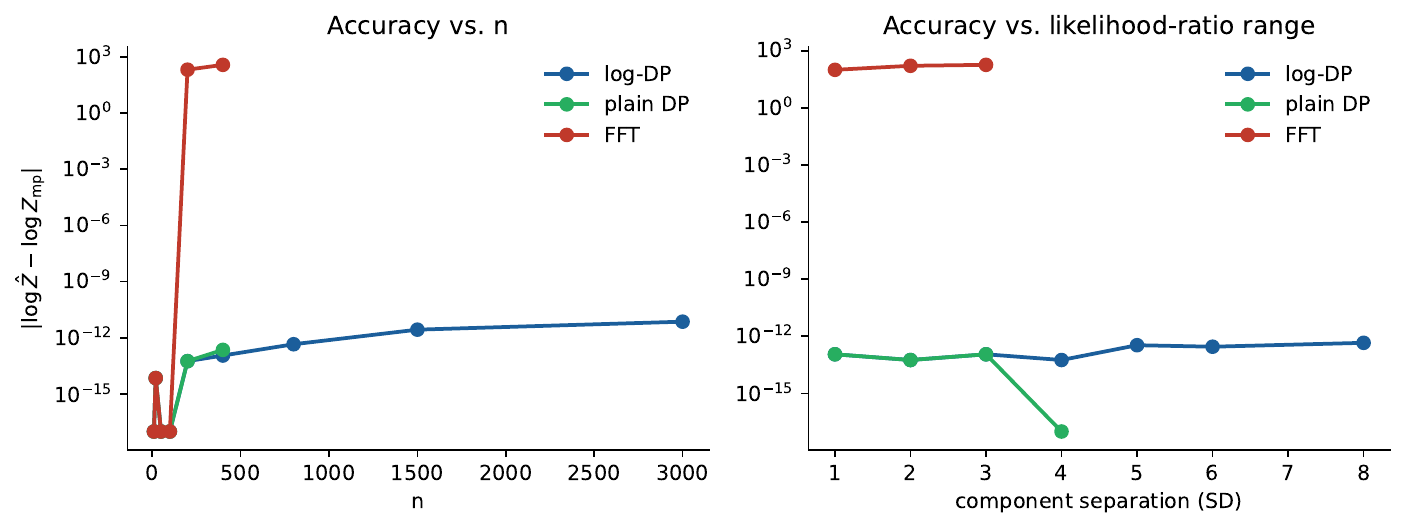}
	\caption{Relative error of $\log Z$ against a $60$-digit exact-arithmetic reference. The log-domain DP holds $\sim\!10^{-12}$ accuracy for all $n$ (left) and all component separations (right). The plain-domain DP stays accurate until it cleanly overflows to \texttt{NaN}; the FFT returns \emph{finite but grossly wrong} values well before overflow, because linear-domain convolution error scales with the largest elementary symmetric polynomial and annihilates the small ones.}
	\label{fig:accuracy}
\end{figure}

\newpage
\pagestyle{fancy}
\fancyhf{}
\rhead{}
\lhead{\leftmark}
\rfoot{\thepage}
\bibliography{refs}
\pagestyle{fancy}
\fancyhf{}
\lhead{\leftmark}
\rfoot{\thepage}
\newpage
\rhead{\leftmark}

\end{document}